\def\comments{0}
\def\nofig{0}

\documentclass[11pt]{article}

\usepackage{preamble}

\newcommand{\tvd}{\mathrm{TV}}
\newcommand{\Set}[1]{\mathcal{#1}} 

\renewcommand{\epsilon}{\varepsilon}

\newcommand{\advg}{\mathrm{adv}}
\newcommand{\LLR}{\operatorname{LLR}}
\newcommand{\sLLR}{\operatorname{sLLR}}
\newcommand{\CLLR}{\operatorname{cLLR}}
\newcommand{\sCLLR}{\operatorname{scLLR}}
\newcommand{\NCLLR}{\operatorname{ncLLR}}
\newcommand{\SC}{\mathit{SC}}

\newcommand{\Lap}{\mathrm{Lap}}

\newcommand{\dx}[1][x]{\, d#1}

\title{The Structure of Optimal Private Tests \\ for Simple Hypotheses}

\author{
	\mbox{Cl{\'e}ment L. Canonne\thanks{Department of Computer Science, Stanford University. Supported by a Motwani Postdoctoral Fellowship. \url{ccanonne@cs.stanford.edu}} \hspace{10pt}}
\and \mbox{Gautam Kamath\thanks{Simons Institute for the Theory of Computing. Supported as a Microsoft Research Fellow, as part of the Simons-Berkeley Research Fellowship program. \url{g@csail.mit.edu}} \hspace{10pt}}
\and \mbox{Audra McMillan\thanks{Department of Computer Science,
    Boston University and College of Computer and Information Science,
    Northeastern University. Research supported by NSF award
    AF-1763786, a Sloan Foundation Research Award, and a postdoctoral fellowship from BU's Hariri Institute
    for Computing.   \url{audram@bu.edu}} \hspace{10pt}}
\and \mbox{Adam Smith\thanks{Department of Computer Science, Boston University. Research supported by NSF awards IIS-1447700 and AF-1763786 and a Sloan Foundation Research Award.   \url{ads22@bu.edu}} \hspace{10pt}}
\and \mbox{Jonathan Ullman\thanks{College of Computer and Information Science, Northeastern University.  Supported by NSF grants CCF-1718088, CCF-1750640, and CNS-1816028, and a Google Faculty Research Award. \url{jullman@ccs.neu.edu}} \hspace{10pt}}
}

\begin{document}

\date{}

\pagenumbering{gobble}
\maketitle

\vspace*{-1cm}
\begin{center}
\emph{In memory of Stephen E.\ Fienberg (1942--2016)}
\end{center}

\begin{abstract}
Hypothesis testing plays a central role in statistical inference, and
is used in many settings where privacy concerns are paramount.  This
work answers a basic question about privately testing simple
hypotheses: given two distributions $P$ and $Q$, and a privacy level
$\eps$, how many i.i.d.~samples are needed to distinguish $P$ from $Q$
subject to $\eps$-differential privacy, and what sort of tests have
optimal sample complexity?  Specifically, we characterize this sample
complexity up to constant factors in terms of the structure of $P$ and
$Q$ and the privacy level $\eps$, and show that this sample complexity
is achieved by a certain randomized and clamped variant of the
log-likelihood ratio test.  Our result is an analogue of the
classical Neyman--Pearson lemma in the setting of private hypothesis
testing.  We also give an application of our result to private change-point detection.
Our characterization applies more generally to hypothesis tests satisfying essentially
any notion of algorithmic stability, which is known to imply strong generalization bounds in adaptive data analysis,
and thus our results have applications even when privacy is not a primary concern.
\end{abstract}

%

\vfill
\newpage

\newcommand{\sayP}{\text{``$P$''}}
\newcommand{\sayQ}{\text{``$Q$''}}
\newcommand{\cH}{{\mathcal{H}}}

\pagenumbering{arabic}
\section{Introduction}

Hypothesis testing plays a central role in statistical inference,
analogous to that of decision or promise problems in computability and
complexity theory. A hypothesis testing problem is specified by two
disjoint sets of probability distributions over the same set, called
hypotheses, $\cH_0$ and $\cH_1$. An algorithm $T$ for this problem, called a \emph{hypothesis test}, is given a sample $x$ from an unknown
distribution $P$, with the requirement that $T(x)$ should, with high probability, output ``0'' if $P\in \cH_0$, and ``1''
if $P\in\cH_1$. There is no requirement for distributions outside of
$\cH_0 \cup \cH_1$.  In computer science, such problems sometimes
go by the name \emph{distribution property testing.}

Hypothesis testing problems are important in their own right, as they
formalize yes-or-no questions about an underlying population based on a
randomly drawn sample, such as whether education strongly influences life
expectancy, or whether a particular medical treatment is
effective. Successful hypothesis tests with high degrees of confidence 
remain the gold standard for publication in top journals in the
physical and social sciences.  Hypothesis testing problems are also
important in the theory of statistics and machine
learning, as many lower bounds for estimation and
optimization problems are obtained
by reducing from hypothesis testing.

This paper aims to understand the structure and sample complexity of
optimal hypothesis tests subject to strong privacy guarantees. Large
collections of personal information are now ubiquitous, but their use
for effective scientific discovery remains limited by concerns about
privacy. In addition to the well-understood settings of data collected
during scientific studies, such as clinical experiments and surveys, many other data sources where privacy concerns are paramount are now being tapped for socially beneficial analysis, such as Social Science One~\cite{SSO}, which aims
to allow access to data collected by Facebook and similar
companies.

We study algorithms that satisfy \emph{differential  privacy (DP)}~\cite{DworkMNS06}, a restriction on the algorithm that
ensures meaningful privacy guarantees against an adversary with
arbitrary side information~\cite{KasiviswanathanS08}. Differential privacy has come
to be the de facto standard for the analysis of private data, used as a
measure of privacy for data analysis systems at Google~\cite{ErlingssonPK14},
Apple~\cite{AppleDP17}, and the U.S.~Census Bureau~\cite{DajaniLSKRMGDGKKLSSVA17}. 
Differential privacy and related distributional notions of
\emph{algorithmic stability} can be crucial for statistical validity
even when confidentiality is not a direct concern, as they provide
generalization guarantees in an adaptive setting~\cite{DworkFHPRR15}.

Consider an algorithm that takes a set of data points from a set
$\cX$---where each point belongs to some individual---and produces some
public output. We say the algorithm is differentially private if no
single data point can significantly impact the distribution on
outputs. Formally, we say two data sets $x, x'\in \cX^n$ of the same size are \emph{neighbors} if they differ in at most one entry.

\begin{definition}[\cite{DworkMNS06}]
  A randomized algorithm $T$ taking inputs in $\cX^*$ and returning
  random outputs in a space with event set $\cS$ is
  \emph{$\eps$-differentially private} if for all $n\geq 1$, for
  all neighboring data sets $x,x'\in \cX^n$, and for all events $S\in
  \cS$,
  $$\pr{}{T(x)\in S}\leq e^\eps \pr{}{T(x')\in S}\,. $$
  For the special case of tests returning output in
  $\bracket{0,1}$, the output distribution is characterized by the
  probability of returning ``1''. Letting $g(x)=\pr{}{T(x)=1}$, we can
  equivalently require that 
  $$\max\paren{\frac{g(x) }{ g(x')} , \frac{1-g(x)}{1-g(x')}} \leq
  e^{\eps} \, . $$
\end{definition}

For algorithms with binary outputs, this definition is essentially
equivalent to all other commonly studied notions of privacy and
distributional algorithmic stability (see ``Connections to Algorithmic
Stability'', below). 

\medskip
\noindent\textbf{Contribution: The Sample Complexity of Private Tests for Simple Hypotheses.}
We focus on the setting of i.i.d.~data and singleton hypotheses $\cH_0,\cH_1$, which are called \emph{simple hypotheses}.  The algorithm is given a sample of $n$ points $x_1,\dots,x_n$ drawn i.i.d. from one of two distributions, $P$ or $Q$,
and attempts to determine which one generated the input.  That is, $\cH_0=\bracket{P^n}$ and $\cH_1 =\bracket{Q^n}$.  We investigate the following question.

\vspace{-5pt}
\begin{quote} \em
  Given two distributions $P$ and $Q$ and a privacy parameter $\eps>0$,
  what is the minimum number of samples (denoted  $\SC^{P,Q}_{\eps}$) needed
  for an $\eps$-differentially private test to reliably distinguish $P$ from $Q$, and 
  what are optimal private tests?
\end{quote}
\vspace{-5pt}

These questions are well understood in the classical, nonprivate
setting. The number of samples needed to distinguish $P$ from $Q$ is
$\Theta(1/H^2(P,Q))$, where $H^2$ denotes the squared Hellinger
distance \eqref{eq:Hellinger}.\footnote{This statement is folklore, but see, e.g.,~\cite{BarYossef02} for the lower bound,~\cite{Canonne17b} or Corollary~\ref{cor:sllr-sc} for the upper bound.}
Furthermore, by the Neyman--Pearson
lemma, the exactly optimal test consists of computing the likelihood
ratio $P^n(x)/Q^n(x)$ and comparing it to some threshold.

We give analogous results in the private setting.  First, we give a
closed-form expression that characterizes the sample complexity up to
universal multiplicative constants, and highlights the range of
$\eps$ in which private tests use a similar amount of data to the
best nonprivate ones. We also give a specific, simple test that achieves that
sample complexity. Roughly, the test makes a noisy decision based on a
``clamped'' log likelihood ratio in which the influence of each data point
is limited. 
The sample complexity has the form
$\Theta(1/\advg_1)$, where $\advg_1$ is the advantage of the test over
random guessing on a sample of size $n=1$.  The optimal test and its sample complexity
are described in Theorem~\ref{thm:private-hypothesis-testing}.

Our result provides the first instance-specific characterization of a statistical problem's complexity for differentially
private algorithms. Understanding the private sample complexity of statistical
problems is delicate.  We know there are regimes where statistical
problems can be solved privately ``for free'' asymptotically
(e.g. \cite{DworkMNS06,ChaudhuriMS11,Smith11,KarwaV18}) and others where there is a significant cost, even for
relaxed definitions of privacy~(e.g. \cite{BunUV14,DworkSSUV15}), and we remain far
from a general characterization of the statistical cost of privacy.
Duchi, Jordan, and Wainwright~\cite{DuchiJW13} give a characterization for the
special case of simple tests by \emph{local} differentially private algorithms, a more restricted setting where
samples are randomized individually, and the test makes a decision
based on these randomized samples. Our characterization in
the general case is more involved, as it exhibits several distinct
regimes for the parameter $\eps$.


Our analysis relies on a number of tools of independent
interest: a characterization of private hypothesis testing in terms of couplings
between distributions on $\cX^n$, and a novel interpretation
of Hellinger distance as the advantage over random guessing
of a specific,
randomized likelihood ratio test.

\medskip
\noindent\textbf{The Importance of Simple Hypotheses.}
Many of the hypotheses that arise in application are not simple, but are so-called
\emph{composite hypotheses}.  For example, deciding if two features are
independent or far from it involves sets $\cH_0$ and $\cH_1$ each
containing many distributions. Yet many of those tests can be reduced
to simple ones.  For example, deciding if the mean of a Gaussian is less than 0 or
greater than 1 can be reduced to testing if the mean is either 0 or
1.  Furthermore, simple tests arise in lower bounds for
estimation---the well-known characterization of parametric estimation
in terms of Fisher information is obtained by showing that the Fisher
information measures variability in the Hellinger distance and then
employing the Hellinger-based characterization of nonprivate simple
tests~(e.g. \cite[Chap. II.31.2, p.180]{Borovkov99}).

Our characterization of private testing implies similar lower bounds
for estimation (along the lines of lower bounds of Duchi and Ruan~\cite{DuchiR18} in the local model of differential privacy).

\medskip
\noindent\textbf{Connection to Algorithmic Stability.}
For hypothesis tests with constant error probabilities, sample
complexity bounds for differential privacy are equivalent, up to
constant factors, to sample complexity bounds for other notions of
distributional algorithmic stability, such as
$(\eps,\delta)$-DP~\cite{DworkKMMN06}, concentrated
DP~\cite{DworkR16,BunS16}, KL- and
TV-stability~\cite{WangLF16,BassilyNSSSU16} (see \cite[Lemma
5]{AcharyaSZ18a}).
(Briefly: if we ensure
that $\Pr(T(x)=1)\in[0.01,0.99]$ for all $x$, then an additive change of
$\eps$ corresponds to an multiplicative change of $1\pm O(\eps)$, and
vice-versa.)
%
%
Consequently, our results imply optimal tests for use in
conjunction with stability-based generalization bounds for adaptive
data analysis, which has generated significant interest in recent years~\cite{DworkFHPRR15,DworkFHPRR-nips-15,DworkFHPRR-science-15,RussoZ16,BassilyNSSSU16,RogersRST16,XuR17,FeldmanS17,FeldmanS18}.

\subsection{Hypothesis Testing}
To put our result in context, we review classical results about non-private hypothesis testing.  Let $P$ and $Q$ be two probability distributions over an arbitrary domain $\mathcal{X}$.  A \emph{hypothesis test} $K \from \cX^* \to \{\sayP,\sayQ\}$ is an algorithm that takes a set of samples $x \in \cX^*$ and attempts to determine if it was drawn from $P$ or $Q$.  
Define the \emph{advantage} of a test $K$ given $n$ samples as
\begin{equation}\label{advantage}
\advg_n(K) = \pr{x \sim P^n}{K(x) = \sayP} - \pr{x \sim Q^n}{K(x) = \sayP}.
\end{equation}
We say that $K$ \emph{distinguishes $P$ from $Q$ with sample complexity $\SC^{P,Q}(K)$} if for every $n \geq \SC^{P,Q}(K)$, $\advg_n(K) \geq 2/3$.
We say $\SC^{P,Q} = \min_{K} \SC^{P,Q}(K)$ is the \emph{sample complexity of distinguishing $P$ from $Q$}.

Most hypothesis tests are based on some real-valued \emph{test statistic} $S \from \cX^* \to \R$ where
\begin{equation*}
K_S(x)= \begin{cases} 
\sayP & \text{if } S(X)\ge\kappa \\ 
\sayQ & \text{otherwise}
\end{cases}
\end{equation*}
for some threshold $\kappa$.  We will sometimes abuse notation and use the test statistic $S$ and the implied hypothesis test $K_S$ interchangeably.

The classical Neyman--Pearson Lemma says that the exact optimal
test\footnote{More precisely, given any test $K$, there is a setting
  of the threshold $\kappa$ for the log-likelihood ratio test that
  weakly dominates $K$, meaning that $\mathbb{P}_{x \sim Q}[\LLR(x) = P] \leq \mathbb{P}_{x
    \sim Q}[K(x) = P]$ and $\mathbb{P}_{x \sim P}[\LLR(x) = Q] \leq \mathbb{P}_{x \sim
    P}[K(x) = Q]$ (keeping the true positive rates $\mathbb{P}_{x \sim P}[K(x) = P]$, $\mathbb{P}_{x \sim Q}[K(x) = Q]$ fixed). One may need to randomize the decision when
  $S(X)=\kappa$ to achieve some tradeoffs between false negative and
  positive rates.} for distinguishing $P,Q$ is the \emph{log-likelihood ratio test} given by the test statistic

\begin{equation}
\LLR(x_1,\dots,x_n) = \sum_{i=1}^{n} \log \frac{P(x_i)}{Q(x_i)}
\,.\label{eq:LLRtest}
\end{equation}
Another classical result says that the optimal sample complexity is characterized by the \emph{squared Hellinger distance} between $P,Q$, which is defined as
\begin{equation} \label{eq:Hellinger}
H^2(P,Q) = \frac{1}{2} \int_{\mathcal{X}} \paren{\sqrt{P(x)} - \sqrt{Q(x)}}^{2}\dx.
\end{equation}
Specifically, $\SC^{P,Q} = \SC^{P,Q}(\LLR) = \Theta(1/H^2(P,Q))$. Note that the same metric provides upper and lower bounds on the sample complexity.

\subsection{Our Results} \label{sec:results}

Our main result is an approximate characterization of the sample complexity of $\eps$-differentially private tests for distinguishing $P$ and $Q$.  Analogous to the non-private case, we will write 
$
\SC^{P,Q}_{\eps} = \min_{\textrm{$\eps$-DP $K$}} \SC^{P,Q}(K)
$
to denote the \emph{sample complexity of $\eps$-differentially privately ($\eps$-DP) distinguishing $P$ from $Q$}, and we characterize this quantity up to constant factors in terms of the structure of $P,Q$ and the privacy parameter $\eps$.
Specifically, we show that a \emph{privatized clamped log-likelihood ratio test} is optimal up to constant factors.
This privatization may be achieved through either the Laplace or Exponential mechanism, and we will prove optimality of both methods.

For parameters $b \geq a$, we define the \emph{clamped log-likelihood ratio statistic},
\begin{equation*} \label{eq:ncllr}
\CLLR_{a,b}(x) = \sum_i \left[ \log \frac{P(x_i)}{Q(x_i)} \right]_{a}^{b},
\end{equation*}
where $[ \cdot ]_{a}^{b}$ denotes the projection onto the interval
$[a,b]$ (that is, $[ z ]_{a}^{b} = \max(a, \min(z, b))$).

Define the \emph{soft clamped log-likelihood test}:
\begin{equation*}
\sCLLR_{a,b}(x)= \begin{cases} 
P & \text{with probability } \propto \exp(\frac12 \CLLR_{a,b}(x))\\ 
Q & \text{with probability }  \propto 1
\end{cases}
\end{equation*}

The test $\sCLLR$ is an instance of the \emph{exponential mechanism}~\cite{McSherryT07}, and thus $\sCLLR_{a,b}$ satisfies $\eps$-differential privacy for $\eps = \frac{b-a}{2}$.  

Similarly, define the \emph{noisy clamped log-likelihood ratio test}:
\begin{equation*}
\NCLLR_{a,b}(X) = 
\begin{cases} 
  P & \text{if } \CLLR_{a,b}(x) + \mathrm{Lap}\left(\frac{1}{\eps(b-a)}\right) > 0\\ 
Q & \text{otherwise}
\end{cases}
\end{equation*}
The test $\NCLLR$ is an instance of postprocessing the Laplace mechanism~\cite{DworkMNS06}, and satisfies $\eps$-differential privacy.

Our main result is that, for every $P,Q$, and every $\eps$, the tests $\sCLLR_{-\eps',\eps}$ and $\NCLLR_{-\eps',\eps}$ are optimal up to constant factors, for some appropriate $0 \leq \eps' \leq \eps$.  To state the result more precisely, we introduce some additional notation.  First define
\begin{equation}\label{def:tau}
\tau = \tau(P,Q) \triangleq \max \left\{\int_{\mathcal{X}} \max\{P(x)-e^{\eps}Q(x),0\}\dx, \int_{\mathcal{X}}\max\{Q(x)-e^{\eps}P(x),0\}\dx \right\},
\end{equation}
and assume without loss
of generality that $\tau=\int_{\mathcal{X}}
\max\{P(x)-e^{\eps}Q(x),0\}\dx$, which we assume for the
remainder of this work.\footnote{For $\alpha \geq 0$, the quantity
  $D_{\alpha}(P \| Q)=\int \max(P(x)-\alpha Q(x),0)\dx$ is an
  $f$-divergence and has appeared in the literature before under the
  names \emph{$\alpha$-divergence}, \emph{hockey-stick divergence}, or
  \emph{elementary divergence}~\cite{LiuCV15,BalleBG18} (for
  $\alpha=1$, one obtains the usual total variation distance). Thus,
  $\tau$ is the maximum of the divergences $D_{e^{\eps}}(P \| Q)$ and
  $D_{e^{\eps}}(Q \| P)$. It can also be described as the smallest
  value $\delta$ such that $P$ and $Q$ are $(\eps,
  \delta)$-indistinguishable~\cite{DworkR14}.} Next, let $0 \leq \eps' \leq \eps$ be the largest value such that
\[
\int_{\mathcal{X}} \max\{P(x)-e^{\eps}Q(x),0\}\dx = \int_{\mathcal{X}}\max\{Q(x)-e^{\eps'}P(x),0\}\dx = \tau,
\]
whose existence is guaranteed by a continuity argument (a formal argument is given in Appendix~\ref{sec:eps-proof}).  We give an illustration of the definition of $\tau$ and $\eps'$ in Figure~\ref{fig:clamp}.  Finally, define $\tilde{P} = \min\{e^{\eps}Q, P\}$ and $\tilde{Q} = \min\{e^{\eps'}P, Q\}$ and normalize by $(1-\tau)$ to obtain distributions 
\begin{equation}\label{def:primedistributions}
P'=\tilde{P}/(1-\tau) \;\; \text{ and } \;\; Q'= \tilde{Q}/(1-\tau).
\end{equation}  
The distributions $P',Q'$ are such that
$$-\eps' \leq \log  \frac{P'(x)}{Q'(x)} \leq \eps \, ,$$
and
$$P = (1-\tau) P' + \tau P'' \quad \text{and} \quad Q = (1-\tau)Q' +
\tau Q'' \, ,$$
where $P''$ and $Q''$ are distributions with disjoint support. 
The quantity $\tau$ is
the smallest possible number for which such a representation is
possible. 
With these definitions in hand, we can now state our main result.
\begin{thm}
\label{thm:private-hypothesis-testing}
For every pair of distributions $P,Q$, and every $\eps > 0$, the optimal sample complexity for $\eps$-differentially private tests is achieved by either the soft or noisy clamped log-likelihood test, and satisfies
  \begin{align*}
    \SC^{P,Q}_{\eps} &= \Theta(\SC^{P,Q}(\NCLLR_{-\eps',\eps})) = \Theta(\SC^{P,Q}(\sCLLR_{-\eps',\eps})) \\
    &= \Theta\left( \frac{1}{\eps \tau(P,Q) + (1-\tau)H^2(P',Q')} \right) = \Theta\left(\frac{1}{\advg_1(\sCLLR_{-\eps',\eps})} \right).
  \end{align*}
\end{thm} 
When $\eps \geq \max_x\left|\log P(x)/Q(x)\right|$, Theorem
\ref{thm:private-hypothesis-testing} reduces to
$\SC^{P,Q}_{\eps}=\Theta\left(\frac{1}{H^2(P,Q)}\right)$, which is the
sample complexity for distinguishing between $P$ and $Q$ in the
non-private setting. This implies that we get privacy for free asymptotically in this parameter regime. We will focus on proving the first equality in this paper, the second is proved in Appendix~\ref{sec:adv-scllr}.

\medskip
\noindent\textbf{Comparison to Known Bounds.} For $\eps<1$, the bounds
$$
\frac{1}{H^2(P,Q)} \leq \SC_{\eps}^{P,Q} \leq O\left(\frac{1}{\eps H^2(P,Q)}\right)
$$
follow directly from the non-private sample complexity.  Namely, the lower bound is the non-private sample complexity and the upper bound is obtain by applying the \emph{sample-and-aggregate} technique~\cite{NissimRS07} to the optimal non-private test.  They can be recovered from Theorem \ref{thm:private-hypothesis-testing} by noting that
\begin{align*}
  \eps H^2(P,Q) = \frac{\eps}{2} \| \sqrt P - \sqrt Q \|^2_2 
  &= O\left(\eps\left\|\sqrt{P}-\sqrt{\tilde{P}}\right\|^2_2+\eps\left\|\sqrt{\tilde{P}}-\sqrt{\tilde{Q}}\right\|^2_2+\eps\left\|\sqrt{\tilde{Q}}-\sqrt{Q}\right\|^2_2\right)\\
&=O( \eps\tau+\eps(1-\tau)H^2(P',Q')+\eps\tau)\\
&= O(\eps\tau+(1-\tau)H^2(P',Q'))
\end{align*}
and 
\begin{align*}
  \eps\tau+(1-\tau)H^2(P',Q')&=\eps \cdot \int_{\Set{S}}|P(x)-e^{\eps}Q(x)|\dx+\|\sqrt{\tilde{P}}-\sqrt{\tilde{Q}}\|_2^2\\
  &\le \eps \cdot \int_{\Set{S}}|P(x)-Q(x)|\dx+\|\sqrt{P}-\sqrt{Q}\|_2^2\\
  &\le \eps \cdot \frac{1+e^{\eps/2}}{e^{\eps/2}-1} \cdot \int_{\Set{S}}(\sqrt{P(x)}-\sqrt{Q(x)})^2\dx+H^2(P,Q)\\
&=O(H^2(P,Q)),
\end{align*}
where $\Set{S} = \setOf{x}{P(x) - e^{\eps}Q(x) > 0}$. 

\begin{figure}[t!] 
\centering
\ifnum\nofig=0
\begin{tikzpicture}    \begin{axis}[domain=0:1,samples=100,axis lines = left, xlabel = $x$, xtick = {0,1}]
        \addplot+[mark=none, name path=P,color=blue] {2*x};
        \addlegendentry{$P(x)$};
        \addplot+[mark=none,color=blue,dashed] {e^(0.2)*2*x};
        \addlegendentry{$e^\eps P(x)$};
        \addplot+[mark=none, name path=Q,color=red] {1/0.614*1/(1+sqrt(x))};
        \addlegendentry{$Q(x)$};
        \addplot+[mark=none,color=red,dashed] {e^(0.2)*1/0.614*1/(1+sqrt(x))};
        \addlegendentry{$e^\eps Q(x)$};
        \addplot+[mark=none, name path=EP,draw=none] {min(e^(0.2)*2*x,1/0.614*1/(1+sqrt(x)))};
        \addplot+[mark=none, name path=EQ,draw=none] {min(e^(0.2)*1/0.614*1/(1+sqrt(x)), 2*x)};
        \addplot[blue!10] fill between[of = P and EQ];
        \addplot[red!10] fill between[of = Q and EP];
    \end{axis}
\end{tikzpicture}
\else
    \textsc{\color{red}FIGURE SHALL APPEAR}
\fi
\caption{An illustration of the definition of $\tau$, for $\eps=0.2$
  and for two densities $P,Q$ over $\mathcal{X}=[0,1]$. The blue
  shaded area represents $\int \max\{ P(x) - e^\eps Q(x), 0 \}\dx$,
  while the red corresponds to $\int \max\{ Q(x) - e^\eps P(x), 0
  \}\dx$. The larger of these two is $\tau(P,Q)$.  If the blue area is larger than the red area, the definition of $\eps'$ corresponds to lowering the dotted blue curve until the two are the same size.}  \label{fig:clamp}
\end{figure}
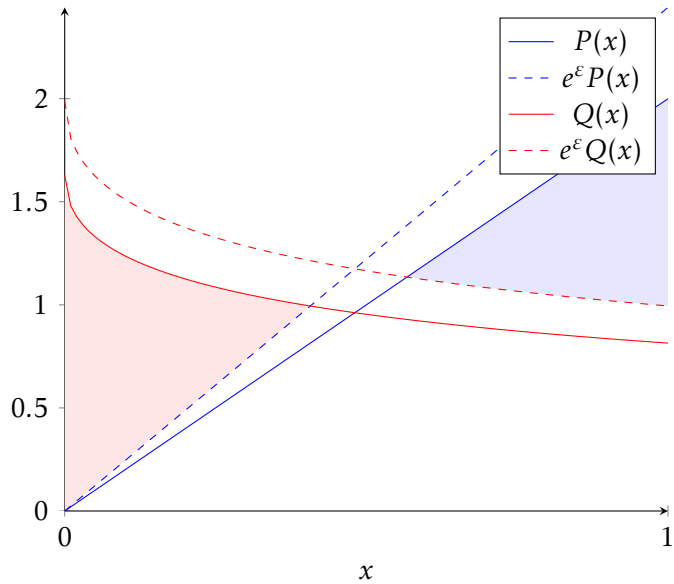

\subsubsection{Application: Private Change-Point Detection}
\label{sec:intro-cpd}
As an application of our result, we obtain optimal private algorithms
for \emph{change-point detection}.  Given distributions $P$ and $Q$,
an algorithm solving \emph{offline change-point detection for $P$ and
  $Q$} takes a stream $x = (x_1,x_2,\dots,x_n) \in \cX^n$ with the
guarantee that the there is an index $k^*$ such that first $k^*$
elements are sampled i.i.d.~from $P$ and the latter elements are
sampled i.i.d.~from $Q$, and attempts to output $\hat{k} \approx k^*$.
We can also consider an online variant where elements $x_i$ arrive one
at a time.

Change-point detection has a long history in statistics and information theory (e.g.~\cite{Shewhart31,Page54,Page55,Shiryaev63,Lorden71,Pollak85, Moustakides86, Pollak87, Lai95, Kulldorff01, Mei06,  VeeravalliB14}). Cummings {\em et al.}~\cite{CummingsKMTZ18} recently gave the first private algorithms for change-point detection.  Their algorithms are based on a private version of the log-likelihood ratio, and in cases where the log-likelihood ratio is not strictly bounded, they relax to a weaker distributional variant of differential privacy.  Using Theorem~\ref{thm:private-hypothesis-testing}, we can achieve the standard worst-case notion of differential privacy, and to achieve optimal error bounds for every $P,Q$.

\begin{thm}[Informal] \label{thm:cpd-intro}
For every pair of distributions $P$ and $Q$, and every $\eps > 0$, there is an $\eps$-differentially private algorithm that solves offline change-point detection for $P$ and $Q$ such that, with probability at least $9/10$, $|\hat{k} - k^*| = O(\SC^{P,Q}_{\eps})$.
\end{thm}

The expected error in this result is optimal up to constant factors for every pair $P,Q$, as one can easily show that the error must be at least $\Omega(\SC^{P,Q}_{\eps})$.  
Theorem~\ref{thm:cpd-intro} can be extended to give an arbitrarily small probability $\beta$ of failure, and can be extended to the online change-point detection problem, although with more complex accuracy guarantees.  
Our algorithm introduces a general reduction from private change-point detection for families of distributions $\cH_0$ and $\cH_1$ to private hypothesis testing for the same family, which we believe to be of independent interest.

\subsection{Techniques}

\noindent\textbf{First Attempts.}
A folklore result based on the sample-and-aggregate paradigm~\cite{NissimRS07} shows that for every $P,Q$ and every $\eps > 0$, $\SC^{P,Q}_{\eps} = O(\frac{1}{\eps} \SC^{P,Q})$, meaning privacy comes at a cost of \emph{at most} $O(\frac{1}{\eps})$.\footnote{See, e.g.,~\cite{CaiDK17} for a proof.}  However, there are many examples where $\SC^{P,Q}_{\eps} = O(\SC^{P,Q})$ even when $\eps = o(1)$, and understanding this phenomenon is crucial.  

A few illustrative pairs of distributions $P,Q$ will serve to demonstrate the difficulties that go into formulating and proving Theorem~\ref{thm:private-hypothesis-testing}.  First, consider the domain $\cX = \{0,1\}$ of size two (i.e.~Bernoulli distributions).  To distinguish $P = \mathrm{Ber}(\frac{1+\alpha}{2})$ from $Q = \mathrm{Ber}(\frac{1-\alpha}{2})$, the optimal non-private test statistic is $S(x) = \sum_{i} x_i$, which requires $\Theta(\frac{1}{ \alpha^2} ) = \Theta( \frac{1}{H^2(P,Q)} )$ samples.

To make the test differentially private, one can use a soft version of the test that outputs \sayP~with probability proportional to $\exp( \eps \sum_i x_i )$ and \sayQ~with probability proportional to $\exp( \eps \sum_i 1-x_i )$.  This private test has sample complexity 
$\Theta(\frac{1}{\alpha^2} + \frac{1}{\alpha \eps}),$
which is optimal.  In contrast, if $P = \mathrm{Ber}(0)$ and $Q = \mathrm{Ber}(\alpha)$, then the non-private sample complexity becomes $\Theta(\frac{1}{\alpha})$, and the optimal private sample complexity becomes $\Theta( \frac{1}{\alpha \eps})$.  In general, for $\cX = \{0,1\}$, one can show that 
\begin{equation} \label{eq:bersc}
\SC^{P,Q}_{\eps} = \Theta\left( \frac{1}{H^2(P,Q)} + \frac{1}{\eps  \mathit{TV}(P,Q)} \right)\,.
\end{equation}

The sample complexity of testing Bernoulli distributions \eqref{eq:bersc} already demonstrates an important phenomenon in private hypothesis testing---for many distributions $P,Q$, there is a ``phase transition'' where the sample complexity takes one form when $\eps$ is sufficiently large and another when $\eps$ is sufficiently small, and often the sample complexity in the ``large $\eps$ regime'' is equal to the non-private complexity up to lower order terms.  A key challenge in obtaining Theorem~\ref{thm:private-hypothesis-testing} is to understand these transitions, and to understand the sample complexity in each regime.

Since each of the terms in \eqref{eq:bersc} is a straightforward lower bound on the sample complexity of private testing, one might conjecture that \eqref{eq:bersc} holds for every pair of distributions.  However, our next illustrative example shows that this conjecture is false even for the domain $\cX = \{0,1,2\}$.  Consider the distributions given by the densities
$$
P = (0, 0.5, 0.5) \quad \textrm{and} \quad Q = (2\alpha^{3/2}, 0.5 + \alpha - \alpha^{3/2}, 0.5 - \alpha - \alpha^{3/2})\,.
$$
For these distributions, the log-likelihood ratio statistic is roughly equivalent to the statistic that counts the number of occurrences of ``0,'' $S_{\{0\}}(x) = \sum_{i} \ind{x_i = 0}$, and has sample complexity $\Theta(\frac{1}{\alpha^{3/2}})$.  For this pair of distributions, the optimal private test depends on the relationship between $\alpha$ and $\eps$.  One such test is to simply use the soft version of $S_{\{0\}}$, and the other is to use the soft version of the test $S_{\{0,1\}} = \sum_{i} \ind{x_i \in \{0,1\}}$ that counts occurrences of the first two elements.  One can prove that the better of these two tests is optimal up to constant factors, giving
$$
\SC^{P,Q}_{\eps} = \Theta\left( \min \left\{ \frac{1}{\alpha^{3/2} \eps}, \frac{1}{\alpha^{2}} + \frac{1}{\alpha \eps} \right\} \right)\,.
$$
For these distributions,~\eqref{eq:bersc} reduces to $\Theta( \frac{1}{\alpha^{3/2}} + \frac{1}{\alpha \eps})$, so these distributions show that the optimal sample complexity can be much larger than~\eqref{eq:bersc}.  Moreover, these distributions exhibit that the optimal sample complexity can vary with $\eps$ in complex ways, making several transitions and never matching the non-private complexity unless $\alpha$ or $\eps$ is constant.

\medskip
\noindent\textbf{Key Ingredients.}
The second example above demonstrates that the optimal test itself can vary with $\eps$ in an intricate fashion, which makes it difficult to construct a single test for which we can prove matching upper and lower bounds on the sample complexity.  The use of the clamped log-likelihood test arose out of an attempt to find a single test that is optimal for the second pair of distributions $P$ and $Q$, and relies on a few crucial technical ingredients.

First, our upper bound on sample complexity relies on a key
observation that the Hellinger distance between two distributions is
exactly the advantage, $\advg_1$, of a soft log-likelihood ratio test
$\sLLR$ on a sample of size 1.
The $\sLLR$ test is a randomized test that outputs $P$ with
probability proportional to $\sqrt{P^n(x)/Q^n(x)} = e^{LLR(x)/2}$,
where $\LLR(x)= \sum_i \log \frac{P(x_i)}{Q(x_i)}$, and $Q$ with
probability proportional to 1.  This characterization of $H^2(P,Q)$ as
the advantage of the $\sLLR$ tester may be of independent interest.

This observation is crucial for our work, because it implies that
$\sLLR$ is $\eps$-DP if $\sup_{x\in \cX} \left|\log \frac{P(x)}{Q(x)}
\right|\le\eps$. That is, in the case that $\sup_{x\in \cX} \left|\log \frac{P(x)}{Q(x)}
\right|\le\eps$,  we get $\eps$-DP \emph{for free} (since Hellinger
distance, and thus $\sLLR$, characterizes the optimal
asymptotic sample complexity). Thus, we use the clamped log-likelihood ratio test, which forces the log-likelihood ratio to be bounded.  Our lower bound in a sense shows that any loss of power in the test due to clamping is \emph{necessary} for differentially private tests.

The lower bound proof proceeds by finding a coupling $\rho$ of $P^n$ and $Q^n$ with low expected Hamming distance $\mathbb{E}_{(X,Y)\sim\rho}[d_H(X,Y)]$, which in turn implies that the sample complexity is large (see e.g.~\cite{AcharyaSZ18a}). The coupling we use essentially splits the support of $P$ and $Q$ into two subsets, those elements with $\log \frac{P(x)}{Q(x)} \ge\eps$ and the remaining elements. To construct the coupling, given a sample $X\sim P^n$, the high-ratio elements for which $\log\frac{P(x)}{Q(x)} \ge\eps$ are resampled with probability related to the ratio. The contribution of this step to the Hamming distance gives us the $\eps\tau$ part of the lower bound. The data set consisting of  the $m \leq n$ elements that have not yet been resampled is then coupled using the total variation distance coupling between $P^{m}$ and $Q^{m}$. Therefore, with probability $1-\tvd(P^{m}, Q^{m})$ this part of the data set remains unchanged. This part of the coupling results in the $H^2(P',Q')$ part of the lower bound. 

The proof of the upper bound also splits into two parts, roughly corresponding to the same aspects of the distributions $P$ and $Q$ as above.  That is, we view our tester as either counting the number of high-ratio elements or computing the log-likelihood ratio on low-ratio elements.  A useful observation is that this duality between the upper and lower bounds is inevitable. In Section \ref{lower}, we characterize the advantage of the optimal tester in terms of Wasserstein distance between $P$ and $Q$ with metric $\min\{\eps d_H(X,Y),1\}$. That is, the advantage of the optimal tester must be matched by some coupling of $P^n$ and $Q^n$. 

\subsection{Related Work}
Early work on differentially private hypothesis testing began in the Statistics community with~\cite{VuS09,UhlerSF13}.  More recently, there has been a significant number of works on differentially private hypothesis testing.  One line of work~\cite{WangLK15,GaboardiLRV16,KiferR17,KakizakiSF17,CampbellBRG18,SwanbergGGRGB19,CouchKSBG19} designs differentially private versions of popular test statistics for testing goodness-of-fit, closeness, and independence, as well as private ANOVA, focusing on the performance at small sample sizes.  
Work by Wang et al.~\cite{WangKLK18} focuses on generating statistical approximating distributions for differentially private statistics, which they apply to hypothesis testing problems.
A recent work by Awan and Slavkovic~\cite{AwanS18} gives a universally optimal test when the domain size is two, however Brenner and Nissim~\cite{BrennerN14} shows that such universally optimal tests cannot exist when the domain has more than two elements.  A complementary research direction, initiated by Cai {\em et al.}~\cite{CaiDK17}, studies the minimax sample complexity of private hypothesis testing.  \cite{AcharyaSZ18a} and~\cite{AliakbarpourDR18} have given worst-case nearly optimal algorithms for goodness-of-fit and closeness testing of arbitrary discrete distributions.  That is, there exists some worst-case distribution $P$ such that their algorithm has optimal sample complexity for testing goodness-of-fit to $P$.  Recently,~\cite{AcharyaKSZ18}~designed nearly optimal algorithms for estimating properties like support size and entropy. 

Another related area~\cite{GaboardiR18,Sheffet18,AcharyaCFT18} studies hypothesis testing in the \emph{local model} of differential privacy.  In particular, Duchi, Jordan, and Wainwright~\cite{DuchiJW13} proved an analogue of our result for the restricted case of locally differentially private algorithms. Their characterization shows that, the optimal sample complexity for $\eps$-DP local algorithms is $\Theta(1/(\eps^2 \mathit{TV}(P,Q)^2))$.  This characterization does not exhibit the same phenomena that we demonstrate in the central model---privacy never comes ``for free'' if $\eps = o(1)$, and the sample complexity does not exhibit different regimes depending on $\eps$.  More generally, local-model tests are considerably simpler, and simpler to reason about, than central-model tests.

There are also several rich lines of work attempting to give tight instance-specific characterizations of the sample complexities of various differentially private computations, most notably \emph{linear query release}~\cite{HardtT10,BhaskaraDKT12,NikolovTZ13,Nikolov15,KattisN17} and \emph{PAC and agnostic learning}~\cite{KasiviswanathanLNRS08,BeimelNS13a,FeldmanX13}.  The problems considered in these works are arguably more complex than the hypothesis testing problems we consider here, the characterizations are considerably looser, and are only optimal up to polynomial factors.

There has been a recent line of work~\cite{DworkFHPRR15,DworkFHPRR-nips-15,DworkFHPRR-science-15,RussoZ16,CummingsLNRW16,BassilyNSSSU16,RogersRST16,FeldmanS17,XuR17,FeldmanS18} on \emph{adaptive data analysis}, in which the same dataset is used repeatedly across multiple statistics analyses, the choice of each analysis depends on the outcomes of previous analyses.  The key theme in these works is to show that various strong notions of algorithmic stability, including differential privacy imply generalization bounds in the adaptive setting.  Our characterization applies to all notions of stability considered in these works.

As an application of our private hypothesis testing results, we provide algorithms for private change-point detection.
As discussed in Section~\ref{sec:intro-cpd}, change-point detection has enjoyed a significant amount of study in information theory and statistics. 
Our results are in the private minimax setting, as recently introduced by Cummings et al.~\cite{CummingsKMTZ18}.
We improve on their results by improving the detection accuracy and providing strong privacy guarantees for all pairs of hypotheses. 

\bigcnote{Add Organization section?}


\section{Upper Bound on Sample Complexity of Private Testing}

In this section, we establish the upper bound part of Theorem~\ref{thm:private-hypothesis-testing}, establishing that the soft clamped log-likelihood ratio test $\sCLLR$ and the noisy clamped log-likelihood ratio test $\NCLLR$ achieve the stated sample complexity. 
In more detail, we begin in Section~\ref{sec:hellinger} by characterizing the Hellinger distance between two distributions as the advantage, $\advg_1$, of a specific randomized test, $\sLLR$. 
Besides being of independent interest, this reformulation also yields some insight on its privatized variant, $\sCLLR$. 
In Section~\ref{sec:nllrt}, we introduce the noisy clamped log-likelihood ratio test ($\NCLLR$), and show that its sample complexity is at least that of $\sCLLR$. 
We then proceed in Section~\ref{sec:ncllr:sc} to upper-bound the sample complexity of $\NCLLR$, which also implies that same bound on that of $\sCLLR$.

Due to the number of named tests in this paper, the reader may find it useful to refer to Appendix~\ref{sec:glossary}, where we enumerate all the tests that we mention.
\subsection{Hellinger Distance Characterizes the Soft Log-Likelihood Test} \label{sec:hellinger}

Recall that the advantage of a test, $\advg_n$, was defined in Equation (\ref{advantage}) and the squared Hellinger distance, $H^2(P,Q)$, between two distributions $P$ and $Q$ is defined as 
\[
H^2(P,Q)= \frac{1}{2} \int_{\mathcal{X}} (\sqrt{P(x)}-\sqrt{Q(x)})^2 \dx,
\]

It has long been known that the Hellinger distance characterizes the asymptotic sample complexity of non-private testing (see, e.g.,~\cite{Borovkov99}).  In this section we show that the Hellinger distance exactly characterizes the advantage of the following randomized test given a single data point:
\[
\sLLR(x) =  
\begin{cases}
P & \text{ with probability } g(x)\\
Q & \text{ with probability } 1-g(x)
\end{cases}
\]
where 
\[
    g(x)=\frac{\exp\left(\frac{1}{2}\log\frac{P(x)}{Q(x)}\right)}{1+\exp\left(\frac{1}{2}\log\frac{P(x)}{Q(x)}\right)}\in[0,1].
\]

Considering the advantage of $\sLLR$ might seem puzzling at first glance, since the classic likelihood ratio test $\LLR$ enjoys a better advantage.
More specifically, the value of $\advg_1$ for these two tests is $H^2(P,Q)$ (Theorem~\ref{hellingertest}) and $TV(P,Q)$ (by the definition of total variation distance), respectively, and $H^2(P,Q) \leq TV(P,Q)$ (see, e.g.,~\cite{GibbsS02}), so it would appear that $\LLR$ is the better test.
There are two relevant features of $\sLLR$ which will be useful.
First, as mentioned before, $\sLLR$ is naturally private if the likelihood ratio is bounded.
Second, a tensorization property of Hellinger distance allows us to easily relate the advantage of the $n$-sample test to the advantage of the $1$-sample test.

\begin{theorem}\label{hellingertest}
For any two distributions $P,Q$, the advantage, $\advg_1$, of $\sLLR$ is $H^2(P,Q)$. 
\end{theorem}
\begin{proof}
Note that we can rewrite 
\begin{align*}
H^2(P,Q) &= \frac12 \int_{\cX} (\sqrt{P(x)}-\sqrt{Q(x)})^2 \dx \\
&= \frac12 \int_{\cX} (P(x)-Q(x))\frac{\sqrt{P(x)}-\sqrt{Q(x)}}{\sqrt{P(x)}+\sqrt{Q(x)}} \\
&= \frac12 \int_{\cX} (P(x)-Q(x))\frac{\sqrt{\frac{P(x)}{Q(x)}}-1}{\sqrt{\frac{P(x)}{Q(x)}}+1}\dx.
\end{align*}
Now, $g(x)=\frac{\sqrt{\frac{P(x)}{Q(x)}}}{\sqrt{\frac{P(x)}{Q(x)}}+1} = \frac{1}{2}\left(\frac{\sqrt{\frac{P(x)}{Q(x)}}-1}{\sqrt{\frac{P(x)}{Q(x)}}+1}+1\right)$, and therefore 
\begin{align*}
H^2(P,Q) &= \int_{\cX} (P(x)-Q(x))g(x)\dx \\ 
&= \ex{x \sim P}{g(x)} - \ex{x \sim Q}{g(x)} \\ 
&= \pr{x \sim P}{\sLLR(x)=P} - \pr{x \sim Q}{\sLLR(x) = P}.
\end{align*} 
Thus, the advantage of $\sLLR$ is $H^2(P,Q)$, as claimed. 
\end{proof}

This tells us the advantage of the test which takes only one sample.
As a corollary, we can derive the sample complexity of distinguishing $P$ and $Q$ using $\sLLR$.

\begin{coro}
  \label{cor:sllr-sc}
  $SC^{P,Q}(\sLLR) = O\left(\frac{1}{H^2(P,Q)}\right)$.
\end{coro}
\begin{proof}
  Our analysis is similar to that of~\cite{Canonne17b}.
  Observe that the test $\sLLR$ which gets $n$ samples from either $P$ or $Q$ is equivalent to the test $\sLLR$ which gets $1$ sample from either $P^n$ or $Q^n$.
  By Theorem~\ref{hellingertest}, we have that the advantage of either (equivalent) test is $\advg = H^2(P^n,Q^n)$.
  We will require the following tensorization property of the squared Hellinger distance:
  $$H^2(P_1 \times \dots \times P_n, Q_1 \times \dots \times Q_n) = 1 - \prod_{i=1}^n (1 - H^2(P_i, Q_i)).$$
  With this in hand, 
  $$\advg = H^2(P^n, Q^n) = 1 - (1 - H^2(P,Q))^n = 1 - \exp\left(n \log (1 - H^2(P,Q))\right) \geq 1 - \exp(-nH^2(P,Q)).$$
  Setting $n = \Omega(1/H^2(P,Q))$, we get $\advg \geq 2/3$, as desired.
\end{proof}

\subsection{The Noisy Log-Likelihood Ratio Test}\label{sec:nllrt}
We now consider the noisy log-likelihood ratio test, which, similar to $\sCLLR_{-\eps',\eps}$, is also $\eps$-differentially private.
\begin{equation*}
\NCLLR_{-\eps',\eps}(X) = \sum_i \left[ \log \frac{P(x_i)}{Q(x_i)} \right]_{-\eps'}^{\eps} + \mathrm{Lap}(2)
\end{equation*}
Here $\mathrm{Lap}(2)$ denotes a \emph{Laplace random variable}, which has density proportional to $\exp(-|z|/2)$.  Readers familiar with differential privacy will note that $\sCLLR$ corresponds to the exponential mechanism, while $\NCLLR$ responds to the \emph{report noisy max} mechanism~\cite{DworkR14}, and thus the two should behave quite similarly.  In particular, we have the following lemma.
\begin{lemma}\label{noisetosoft}
For any $P$ and $Q$: 
\begin{enumerate}
\item $SC^{P,Q}(\NCLLR_{-\eps',\eps})= \Omega(SC^{P,Q}(\sCLLR_{-\eps',\eps}))$.
\item Furthermore, if $-\eps' \leq \log\frac{P}{Q} \leq \eps$ then $SC^{P,Q}(\NCLLR_{-\eps',\eps})=\Theta(SC^{P,Q}(\sCLLR_{-\eps',\eps}))$.
\end{enumerate}
\end{lemma}

\begin{proof}
Recall that \[\sCLLR_{-\eps',\eps}(X) = 
\begin{cases}
P & \text{ with probability } g_{\eps}(X)\\
Q & \text{ with probability } 1-g_{\eps}(X)
\end{cases}\]
where $g_{\eps}(X) = \frac{e^{\frac{1}{2}\CLLR_{-\eps',\eps}(X)}}{1+e^{\frac{1}{2}\CLLR_{-\eps',\eps}(X)}}$. If we let the threshold $\kappa=0$ then the test based on the test statistic $\NCLLR_{-\eps',\eps}$ is 
\[\NCLLR_{-\eps',\eps}(X) = 
\begin{cases}
P & \text{ with probability } h_{\eps}(X)\\
Q & \text{ with probability } 1-h_{\eps}(X)
\end{cases}\]
where \[
h_{\eps}(X) =
    \begin{cases}
    1-\frac{1}{2}e^{-\frac{\CLLR_{-\eps',\eps}(X)}{2}} & \text{ if } \CLLR_{-\eps',\eps}(X)>0\\
    \frac{1}{2}e^{\frac{\CLLR_{-\eps',\eps}(X)}{2}} & \text{ if } \CLLR_{-\eps',\eps}(X)<0
    \end{cases}.
\]
Now, we will use the following two inequalities:
\[
  \frac{1}{2}e^{\frac{x}{2}}~\le~\frac{e^{\frac{1}{2}x}}{1+e^{\frac{1}{2}x}}~\le~2\frac{1}{2}e^{\frac{x}{2}} \text{ if } x<0 \;\; \text{ and } \;\; \frac{8}{9}(1-\frac{1}{2}e^{-\frac{x}{2}})\le\frac{e^{\frac{1}{2}x}}{1+e^{\frac{1}{2}x}}\le 1-\frac{1}{2}e^{-\frac{x}{2}} \text{ if } x>0.
\]  
Therefore, noting that 
\[
    \pr{X \sim P^n}{\sCLLR_{-\eps',\eps}(X)=P} = \mathbb{E}_{P^n}[g_{\eps}(X)] \;\; \text{ and } \;\; \pr{X \sim P^n}{\NCLLR_{-\eps',\eps}(X)=P} = \mathbb{E}_{P^n}[h_{\eps}(X)]
\]
are the probabilities of success, we have 
\[\frac{8}{9}\mathbb{E}_{P^n}[h_{\eps}(X)]\le \mathbb{E}_{P^n}[g_{\eps}(X)]\le 2 \mathbb{E}_{P^n}[h_{\eps}(X)] \;\;\text{ and }\;\; \frac{8}{9}\mathbb{E}_{Q^n}[h_{\eps}(X)]\le \mathbb{E}_{Q^n}[g_{\eps}(X)]\le 2 \mathbb{E}_{Q^n}[h_{\eps}(X)].\]
Therefore, if $\NCLLR$ has a probability of success of $5/6$ then $\sCLLR$ has a probability of success of $2/3$. This implies that $SC^{P,Q}(\NCLLR)\ge SC^{P,Q}(\sCLLR)$. 

If $\log\frac{P}{Q} \in [-\eps',\eps]$ then $\CLLR_{-\eps',\eps}(X)=\LLR(X)$ so we have $P^n(X)>Q^n(X)$ iff $\LLR(X)>0$ iff $g_{\eps}(X)\le h_{\eps}(X)$ and therefore
\begin{align*}
\mathbb{E}_{P^n}[g_{\eps}]-\mathbb{E}_{Q^n}[g_{\eps}]&= \int (P^n(X)-Q^n(X))g_{\eps}(X) \dx[X]\\
&\le \int (P^n(X)-Q^n(X))h_{\eps}(X) \dx[X]\\
&= \mathbb{E}_{P^n}[h_{\eps}]-\mathbb{E}_{Q^n}[h_{\eps}],
\end{align*}
which completes the proof.
\end{proof}

\begin{coro}\label{Noptimalimpliessoptimal}
If $\NCLLR$ has asymptotically optimal sample complexity then $\sCLLR$ has asymptotically optimal sample complexity. 
\end{coro}

\subsection{The Sample Complexity of $\NCLLR$}\label{sec:ncllr:sc}

In this section we prove the upper bound in Theorem~\ref{thm:private-hypothesis-testing} for the case where $\tvd(P,Q) < 1$ (i.e.,~the supports of $P$ and $Q$ have non-empty intersection).  This assumption ensures that $\tilde{P},\tilde{Q} \neq 0$, so that $P',Q'$ are well defined.  A proof of the case where $\tvd(P,Q)=1$ is contained in Appendix~\ref{app:disjoint:supports}.  In order to prove the upper bound in Theorem~\ref{thm:private-hypothesis-testing}, we restate it as follows.

\begin{theorem}\label{upperboundforproof}
The $\NCLLR_{-\eps',\eps}$ test is $\eps$-DP and 
\[
SC^{P,Q}(\NCLLR_{-\eps',\eps}) \le
O\left(\frac{1}{\epsilon\tau+(1-\tau)H^2(P',Q')}\right)
=O\left(\min\left\{\frac{1}{\eps\tau},
    \frac{1}{(1-\tau) H^2(P',Q')}\right\}\right) \,.
\]
\end{theorem}

Theorem~\ref{upperboundforproof} combined with a matching lower bound (given later in Theorem~\ref{lowerbound}) imply that $\NCLLR_{-\eps'}^{\eps}$ has asymptotically optimal sample complexity. Thus, by Corollary~\ref{Noptimalimpliessoptimal}, $\sCLLR_{-\eps'}^{\eps}$ has asymptotically optimal sample complexity.

Before proving the bound, we pause to provide some intuition for its
form. As discussed in the introduction, we can write $P$ and $Q$ as
mixtures $P= (1-\tau)P' + \tau P''$ and $Q= (1-\tau)Q' + \tau Q''$
where $P'', Q''$ have disjoint support. Now consider a thought
experiment, in which the test that must
distinguish $P$ from $Q$ using a sample of size $n$ is given, along with the sample $x$, a
list of binary labels $b_1,b_2,...,b_n$ that indicate for each record
whether it was sampled from the first component of the mixture (either
$P'$ or $Q'$), or the second component (either $P''$ or
$Q''$). Of course this can only be a though experiment---these labels are
not available to a real test.

Because the mixture weights are the same for both $P$ and $Q$,
the number of labels of each type would be distributed the same under
$P$ and under $Q$, and so the tester would be faced with two independent
testing problems: distinguishing $P''$ from $Q''$ using a sample of
size about $\tau n$, and distinguishing $P'$ from $Q'$ using a sample
of size about $(1-\tau) n$. It would suffice for the tester to solve
either of these problems.

Theorem~\ref{upperboundforproof} shows that the real tests ($\sCLLR$ and
$\NCLLR$) do as well
as the hypothetical tester that has access to the labels. The two
arguments to the minimum in the theorem statement correspond directly
to the $\eps$-DP sample complexity of distinguishing $P''$ from $Q"$
(which requires $n\tau \geq 1/\eps$) or distinguishing $P'$ from $Q'$
(which requires $n(1-\tau) \geq 1/H^2(P',Q')$).
The
proof proceeds by breaking the clamped log-likelihood ratio into
two pieces, each corresponding to one of the two mixture components
(again, this decomposition is not known to the algorithm). These two
pieces correspond to the test statistics of the optimal testers for
the two separate sub-problems in the thought experiment. We show that
the test does well at distinguishing $P$ from $Q$ as long as either of
these pieces is sufficiently informative.

On a more mechanical level, our proof of Theorem~\ref{upperboundforproof} bounds the expectation and
standard deviation of the two pieces of the test statistic. We use the
following simple lemma, which states that a test statistic $S$
performs well if the distribution of the test statistic on $P$ and
$Q$, $S(P)$ and $S(Q)$, must not overlap too much. The proof is a
simple application of Chebyshev's inequality.


\begin{lemma}[Sufficient Conditions for Noisy Test Statistics]\label{lem:vargap}
Given a function $f\from\mathcal{X}\to\R$, constant $c>0$ and $n>0$, if the test statistic $S(X)=\sum_{i}f(x_i)$ satisfies 
    \[
    \max\left\{\sqrt{\var{P^n}{S(X)}}, \sqrt{\var{Q^n}{S(X)}}\right\}\le c|\ex{P^n}{S(X)}-\ex{Q^n}{S(X)}|
    \]
    then $S$ can be used to distinguish between $P$ and $Q$ with probability of
    success 2/3 and sample complexity at most $n'=12c^2n$.

\end{lemma}

\noindent This result follows by Chebyshev's inequality. For completeness, it is proved in Appendix~\ref{vargapproof}.\smallskip

The definitions of $\tilde{P}$ and $\tilde{Q}$ lend naturally to consider a partition of the space $\cX$, depending on which quantities achieve the minimum in $\min(e^\eps Q,P)$ and $\min(e^{\eps'}P,Q)$. This partition will itself play a crucial role in both the proof of the theorem, and later in our lower bound: accordingly, define
\begin{equation}\label{eq:def:sets:STA}
\Set{S} = \setOf{x}{P(x) - e^{\eps}Q(x) > 0} \quad \textrm{and} \quad \Set{T} = \setOf{x}{Q(x) - e^{\eps'}P(x) > 0}
\end{equation}
and set $\Set{A} = \cX \setminus (\Set{S} \cup \Set{T})$.



\begin{proof}[Proof of Theorem~\ref{upperboundforproof}]
Observe that for all $x\in \Set{A}$, $\tilde{P}(x)=P(x)$ and $\tilde{Q}(x)=Q(x)$ (so that $P'(x)/Q'(x) = P(x)/Q(x)$), that for all $x\in \Set{S}$,   $\log(P'(x)/Q'(x))=\eps$ and that for all $x\in \Set{T}$, $\log(P'(x)/Q'(x))=-\eps'$. To show that the test works, we
first show that the clamped log likelihood ratio (without noise) is a useful
test statistic. In order to apply Lemma \ref{lem:vargap}, we first
calculate the difference $\Delta_{\rm gap}$ in the expectations of $\CLLR_{-\eps',\eps}$ under $P$
and $Q$. For the remainder of the proof, we omit the $-\eps',\eps$
subscript (since clamping always occurs to the same interval).
\begin{align*}
\Delta_{\rm gap}
&= \ex{P^n}{\CLLR(X)}-\ex{Q^n}{\CLLR(X)}\\
&= n\int_{\cX}(P(x)-Q(x))\log\frac{P'(x)}{Q'(x)} \dx\\
&= n(P(\Set{S})-Q(\Set{S}))\eps + n\int_{\Set{A}}(\tilde{P}(x)-\tilde{Q}(x))\log\frac{P'(x)}{Q'(x)} \dx + n(Q(\Set{T})-P(\Set{T}))\eps'\\
&= n(\tilde{P}(\Set{S})-\tilde{Q}(\Set{S})+\tau)\eps+n\int_{\Set{A}}(\tilde{P}(x)-\tilde{Q}(x))\log\frac{P'(x)}{Q'(x)} \dx + n(\tilde{Q}(\Set{T})-\tilde{P}(\Set{T})+\tau)\eps'
\end{align*}
where the last equality follows by the definition of $\tau$.
Moreover, 
\begin{align*}
n \kl{P'}{Q'}&+n \kl{Q'}{P'}\\
&= n \int_{\cX} (P'(x)-Q'(x))\log\frac{P'(x)}{Q'(x)}\dx\\
&= n (P'(\Set{S})-Q'(\Set{S}))\eps + n \int_{\Set{A}} (P'(x)-Q'(x))\log\frac{P'(x)}{Q'(x)}\dx + n (Q'(\Set{T})-P'(\Set{T})) \eps'\\
&= \frac{n}{1-\tau} \left( (\tilde{P}(\Set{S})-\tilde{Q}(\Set{S}))\eps + \int_{\Set{A}} (\tilde{P}(x)-\tilde{Q}(x))\log\frac{P'(x)}{Q'(x)} + (\tilde{Q}(\Set{T})-\tilde{P}(\Set{T})) \eps' \right)\,.
\end{align*}
Therefore
\[
\Delta_{\rm gap} = (1-\tau) n (\kl{P'}{Q'}+\kl{Q'}{P'}) + n\tau\eps +
  n\tau\eps'\ge 2n\left((1-\tau)H^2(P', Q')+\tau\eps\right)
\]
where the last inequality follows since $H^2(P, Q)\le \kl{P}{Q}$ for
any distributions $P$ and $Q$.

We now turn to bounding the variance of the noisy clamped LLR under $P$ and $Q$. Noting that $\var{P^n}{\NCLLR} = \var{P^n}{\CLLR}+8$, by Lemma \ref{lem:vargap} it suffices to show that \[\max\{\var{Q^n}{\CLLR}+8,
\var{P^n}{\CLLR}+8\}\le O(\Delta_{\rm gap}^2),\] or equivalently, \[\Delta_{\rm gap} = \Omega(1) \;\;\;\; \text{and} \;\;\;\; \max\{\var{Q^n}{\CLLR},
\var{P^n}{\CLLR}\}\le O(\Delta_{\rm gap}^2).\]
Recall that $P''$ is a distribution such that $P = \tau P''+(1-\tau)P'$ and the support of $P''$ is contained in $\Set{S}$. Thus,
\begin{align*}
\var{P^n}{\CLLR} &\le n \int_{\cX} P(x)\left(\log\frac{P'(x)}{Q'(x)}\right)^2 \dx\\
&=n \int_{\cX}(\tau P''(x) + (1-\tau)P'(x))\left(\log\frac{P'(x)}{Q'(x)}\right)^2 \dx\\
&=n \left(\tau P''(\Set{S}) \eps^2 + (1-\tau)\int_{\cX} P'(x)\left(\log\frac{P'(x)}{Q'(x)}\right)^2 \dx\right)
\end{align*}
Since $\log\frac{P'(x)}{Q'(x)}\leq\eps\leq 1$, $\lvert \log\frac{P'(x)}{Q'(x)}\rvert \leq 3\cdot \lvert 1-\sqrt{\frac{Q'(x)}{P'(x)}}\rvert$. Therefore, 
\begin{align*}
\int_{\cX} P'(x)\log^2\frac{P'(x)}{Q'(x)} \dx &\leq 9 \int_{\cX} P'(x) \left(1-\sqrt{\frac{Q'(x)}{P'(x)}}\right)^2\dx 
= 18 \cdot H^2(P', Q')\,.
\end{align*}
Also, $P''(\Set{S}) = 1$ and thus,
\begin{align*}
\var{P^n}{\CLLR} &= O(n\tau\eps^2 + (1-\tau)nH^2(P',Q'))=O(n\tau\eps+(1-\tau)nH^2(P',Q')
\end{align*}
Similarly,
\[\var{Q^n}{\CLLR}\le O(n\tau\eps'^2+(1-\tau)H^2(P', Q'))\le O(n\tau\eps+(1-\tau)nH^2(P',Q'))\]

Finally,
\begin{align*}
\max\{\var{Q^n}{\CLLR}, \var{P^n}{\CLLR}\} &= O\left(\frac{(n\tau\eps + n(1-\tau)H^2(P',Q'))^2}{n\tau\eps + n(1-\tau)H^2(P',Q')}\right)\\
&= O\left(\frac{\Delta_{\rm gap}^2}{n(\tau\eps + (1-\tau)H^2(P',Q'))}\right)
\end{align*}
Therefore, if $n \geq \frac{C}{\tau\eps +
    (1-\tau)H^2(P',Q'))}$ for some suitably large constant $C>0$, the above implies that $\max\{\var{Q^n}{\CLLR},
\var{P^n}{\CLLR}\}\le O(\Delta_{\rm gap}^2)$ and $\Delta_{\rm gap}=\Omega(1)$, which concludes our proof.
\end{proof}



\section{Lower Bound on the Sample Complexity of Private Testing}\label{lower}

We now prove the lower bound in Theorem~\ref{thm:private-hypothesis-testing}.  We do so by constructing an appropriate \emph{coupling} between the distributions $P^n$ and $Q^n$, which implies lower bounds for privately distinguishing $P^n$ from $Q^n$.
This style of analysis was introduced in~\cite{AcharyaSZ18a}, though we require a strengthening of their statement.
Specifically, the lower bound of Acharya et al.\ involves $d'_\eps(X,Y) = \eps d_H(X,Y)$, whereas we have $d_\eps(X,Y) = \min(\eps d_H(X,Y),1)$.

For $X,Y \in \mathcal{X}^n$, let $d_{H}(X,Y)$ be the Hamming distance between $X,Y$ (i.e. $| \setOf{i}{x_i \neq y_i} |$).  Given a metric $d\colon\mathcal{X}^n~\times~\mathcal{X}^n~\to~\R_{\geq 0}$ we define the Wasserstein distance $W_{d}(P,Q)$ by
\[
W_{d}(P,Q) = \inf_{\rho} \ex{(X,Y) \sim \rho}{d(X,Y)}
\]
where the $\inf$ is over all couplings $\rho$ of $P^n$ and $Q^n$. Let $d_{\eps}(X,Y) = \min\{ \eps d_{H}(X,Y), 1\}$.

\begin{lem}\label{dpexpectations}
For every $\eps$-DP algorithm $M\from\mathcal{X}^n\to\zo$, if $X$ and $Y$ are neighboring datasets then \[\mathbb{E}[M(X)]\le e^{\eps}\mathbb{E}[M(Y)],\] where the expectations are over the randomness of the algorithm $M$.
\end{lem}

\begin{lem} \label{lem:dp->wasserstein}
  For every $\eps$-DP algorithm $M \from \mathcal{X}^n \to \{P,Q\}$ the advantage satisfies \[\left|\pr{X \sim P}{M(X)=P}-\pr{X \sim Q}{M(X)=P}\right| \leq O(W_{d_{\eps}}(P,Q))\]
\end{lem}

\begin{proof}
Let $\rho\from\mathcal{X}^n\times\mathcal{X}^n\to\R_{\geq 0}$ be a coupling of $P^n$ and $Q^n$ and $M$ be an $\eps$-DP algorithm. We have 
\begin{align*}
  \pr{X \sim P}{M(X)=P} - \pr{X \sim Q}{M(X)=P}&=\int_{\mathcal{X}^n}\int_{\mathcal{X}^n} \left( \rho(X,Y)\pr{M}{M(X)=P}-\rho(X,Y)\pr{M}{M(Y)=P} \right)\dx[X]\dx[Y]\\
  &\le \int_{\mathcal{X}^n}\int_{\mathcal{X}^n} \rho(X,Y)\min\{1,(e^{\eps d_H(X,Y)}-1)\pr{M}{M(Y)=P}\}\dx[X]\dx[Y]\\
&\le 2\int_{\mathcal{X}^n}\int_{\mathcal{X}^n} \rho(X,Y)\min\{1,\eps d_H(X,Y)\}\dx[X]\dx[Y]\\
&= O(\ex{(X,Y) \sim \rho}{d_{\eps}(X,Y)}),
\end{align*}
  where $\pr{M}{\cdot}$ denotes that the probability is over the randomness of the algorithm $M$, and the first inequality follows from Lemma \ref{dpexpectations}.
\end{proof}

The upper bound in Lemma \ref{lem:dp->wasserstein} is in fact tight. We state the converse below for completeness, although we will not use it in this work. The proof of Lemma \ref{lem:wasserstein->dp} is contained in Appendix~\ref{sec:wasserstein->dp}.

\begin{lem} \label{lem:wasserstein->dp}
There is a $\eps$-DP algorithm $M \from \mathcal{X}^n \to \{P,Q\}$ such that
  $$\left|\pr{X \sim P}{M(X)=P}-\pr{X \sim Q}{M(X)=P}\right| \geq \Omega(W_{d_{\eps}}(P,Q)).$$
\end{lem}

We will also rely on the following standard fact characterizing total variation distance in terms of couplings:
\begin{fact}
Given $P$ and $Q$, there exists a coupling $\rho$ of $P$ and $Q$ such that $\pr{\rho}{X\neq Y}=\tvd(P,Q)$. We will refer to this coupling as the \emph{total variation coupling} of $P$ and $Q$.
\end{fact}

We can now prove the lower bound component of Theorem \ref{thm:private-hypothesis-testing}. Recall that $P'$ and $Q'$ were defined in Equation \eqref{def:primedistributions}.

\begin{theorem}\label{lowerbound}
Given $P$ and $Q$, every $\eps$-DP test $K$ that distinguishes $P$ and $Q$ has the property that \[SC^{P,Q}(K)=\Omega\left( \frac{1}{\eps\tau+H^2(P',Q')}\right).\] 
\end{theorem}

\begin{proof}
Consider the following coupling of $P^n$ and $Q^n$: Given a sample $X\sim P^n$, independently for all $x_i\in \Set{S}$\footnote{Recall the definitions of $\Set{S},\Set{T}$ and $\Set{A}$ from Section~\ref{sec:ncllr:sc} (Equation~\eqref{eq:def:sets:STA}).} label the point $\mathbbm{1}$ with probability 
$\frac{e^{\eps}Q(x_i)}{P(x_i)}$, otherwise label it $\mathbbm{2}$. Label all the points in $\Set{A}\cup \Set{T}$ as $\mathbbm{1}$. (In particular, this implies that each $x_i$ is labeled $\mathbbm{1}$ with probability $1-\tau$, and $\mathbbm{2}$ with probability $\tau$.)  Each point labeled $\mathbbm{2}$ is then independently re-sampled from $\Set{T}$ with probability distribution $\frac{Q-e^{\eps'}P}{\tau}\mathbbm{1}_{\Set{T}}$. Let $\Lambda\subseteq[n]$ be the set of points labeled $\mathbbm{1}$, and $n'$ be its size; and note that this set is distributed according to $(P')^{n'}$. We transform this set to a set distributed by $(Q')^{n'}$ using the TV-coupling of $(P')^{n'}$ and $(Q')^{n'}$. The result is a sample from $Q^n$.

Now, we can rewrite
\[
    d_H(X,Y) 
    = \sum_{i\notin\Lambda} \mathbbm{1}_{\{X_i=Y_i\}} + \sum_{i\in\Lambda} \mathbbm{1}_{\{X_i=Y_i\}}
    = d_H(X_{\bar{\Lambda}},Y_{\bar{\Lambda}}) + \mathbbm{1}_{\{X_\Lambda=Y_\Lambda\}}\cdot \sum_{i\in\Lambda} \mathbbm{1}_{\{X_i=Y_i\}}
\] 
and therefore
\begin{align*}
\ex{}{\min\{\eps d_H(X,Y),1\}}
 &= \ex{}{\min\{\eps d_H(X,Y),1\}\mathbbm{1}_{\{X_\Lambda=Y_\Lambda\}}} + \ex{}{\min\{\eps d_H(X,Y),1\}\mathbbm{1}_{\{X_\Lambda\neq Y_\Lambda\}}} \\
 &\leq \eps\ex{}{d_H(X,Y)\mathbbm{1}_{\{X_\Lambda=Y_\Lambda\}}} + \ex{}{\mathbbm{1}_{\{X_\Lambda\neq Y_\Lambda\}}}\\
 &= \eps\ex{}{d_H(X_{\bar{\Lambda}},Y_{\bar{\Lambda}})} + \pr{}{X_\Lambda\neq Y_\Lambda}\,.
\end{align*}
Recalling now that the distribution of  $(X_\Lambda,Y_\Lambda)$ is that of the TV-coupling of $(P')^{n'}$ and $(Q')^{n'}$, and that $\lvert \bar{\Lambda}\rvert = n-n'$, we get
\begin{align*}
\ex{}{\min\{\eps d_H(X,Y),1\}}&\le \ex{}{\eps (n-n')+\tvd( (P')^{n'},(Q')^{n'} ) }\\
&\le \eps \tau n+\ex{}{\sqrt{n'}}H(P',Q')\\
&\le \eps\tau n +\sqrt{(1-\tau)n} \cdot H(P',Q') 
\end{align*}
Therefore, by Lemma \ref{lem:dp->wasserstein}, we have that for every $\eps$-DP test $M$, 
  \[\left|\pr{X \sim P}{M(X)=P}-\pr{X \sim Q}{M(X)=P}\right|\le\eps\tau n +\sqrt{(1-\tau)n} \cdot H(P',Q').\]
Thus, in order for the probability of success to be $\Omega(1)$, we need either $\eps\tau n$ or $\sqrt{(1-\tau)n}H(P',Q')$ to be $\Omega(1)$. That is, $n\ge\Omega\left(\min\left\{\frac{1}{\eps\tau}, \frac{1}{(1-\tau)H^2(P',Q')}\right\}\right)$.
\end{proof}

\section{Application: Differentially Private Change-Point Detection}
\label{sec:cpd}
In this section, we give an application of our method to differentially private change-point detection.
In the change-point detection problem, we are given a time-series of data.
Initially, it comes from a known distribution $P$, and at some unknown time step, it begins to come from another known distribution $Q$.
The goal is to approximate when this change occurs.
More formally, we have the following definition.
\begin{definition}
In the \emph{offline change-point detection problem}, we are given distributions $P,Q$ and a data set $X = \{x_1, \dots, x_n\}$.
We are guaranteed that there exists $k^\ast \in [n]$ such that $x_1, \dots, x_{k^\ast-1} \sim P$ and $x_{k^\ast}, \dots, x_n \sim Q$.
The goal is to output $\hat k$ such that $|\hat k - k^\ast|$ is small. 

In the \emph{online change-point detection problem}, we are given distributions $P,Q$, a stream of data points $X = \{x_1, \dots\}$.
We are guaranteed that there exists $k^\ast$ such that $x_1, \dots, x_{k^\ast-1} \sim P$ and $x_{k^\ast}, \dots \sim Q$.
The goal is to output $\hat k$ such that $|\hat k - k^\ast|$ is small. 
\end{definition}

We study the parameterization of the private change-point detection problem recently introduced by Cummings {\em et al.}~\cite{CummingsKMTZ18}.
\begin{definition} [Change-Point Detection]
An algorithm for a (online) change-point detection problem is $(\alpha,\beta)$-accurate if for any input dataset (data stream), with probability at least $1 - \beta$ outputs a $\hat k$ such that $|\hat k - k^\ast| \leq \alpha$, where the probability is with respect to the randomness in the sampling of the data set and the random choices made by the algorithm.
\end{definition}
\jnote{I took out the part of the def about DP, since it seems obvious what it means for the algorithm to be DP, and is anyway orthogonal to the definition of accuracy.}

Our main result is the following:
\begin{theorem}
\label{thm:cpd}
There exists an efficient $\ve$-differentially private and $(\alpha, \beta)$-accurate algorithm for offline change-point detection from distribution $P$ to $Q$ with 
$$\alpha =  \Theta\left(\frac{1}{\eps \tau(P,Q) + H^2(P',Q')} \cdot \log (1/\beta)\right).$$

Furthermore, there exists an efficient $\ve$-differentially private and $(\alpha, \beta)$-accurate algorithm for online change-point detection from distribution $P$ to $Q$ with the same value of $\alpha$.  This latter algorithm also requires as input a value $n$ such that $n = \Omega\left(\SC^{P,Q}_{\eps} \cdot \log\left(\frac{k^\ast}{n \beta}\right)\right)$.
If the algorithm is accurate, it will observe at most $k^\ast + 2n$ data points, and with high probability observe $k^\ast + O(n \log n)$ data points.
\end{theorem}

For constant $\beta$, the accuracy of our offline algorithm is optimal up to constant factors, since one can easily show that the best accuracy achievable is $\Omega(\SC_{\eps}^{P,Q})$.
A similar statement holds for our online algorithm when the algorithm is given an estimate $n$ of $k^\ast$ such that $n = \poly(k^\ast)$.

As one might guess, this problem is intimately related to the hypothesis testing question studied in the rest of this paper. 
Indeed, our change-point detection algorithm will use the hypothesis testing algorithm of Theorem~\ref{thm:private-hypothesis-testing} as a black box, in order to reduce to a simpler Bernoulli change-point detection problem (see Lemma~\ref{lem:reduction} in Section~\ref{sec:reduction-cpd}).
We then give an algorithm to solve this simpler problem (Lemma~\ref{lem:bernoulli-cpd} in Section~\ref{sec:bernoulli-cpd}), completing the proof of Theorem~\ref{thm:cpd}.
In Section~\ref{sec:identity-cpd}, we show that our reduction is applicable more generally, as we describe an algorithm change-point detection in a goodness-of-fit setting.

\subsection{A Reduction to Bernoulli Change-Point Detection}
\label{sec:reduction-cpd}
In this section, we provide a reduction from private change-point detection with arbitrary distributions to non-private change-point detection with Bernoulli distributions.
\begin{lemma}
\label{lem:reduction}
Suppose there exists an $(\alpha, \beta)$-accurate algorithm which can solve the following restricted change-point detection problem: we are guaranteed that there exists $\tilde k^\ast$ such that $z_1, \dots, z_{\tilde k^\ast-1} \sim 2 \operatorname{Ber}(\tau_0) - 1$ for some $\tau_0 > 2/3$, $z_{\tilde k^\ast+1} , \dots \sim 2\operatorname{Ber}(\tau_1) - 1$ for some $\tau_1 < 1/3$, and $z_{\tilde k^\ast} \in \{\pm 1\}$ is arbitrary.

Then there exists an $\ve$-differentially private and $((\alpha + 1) \cdot \SC^{P,Q}_{\eps} ,\beta)$-accurate algorithm which solves the change-point detection problem, where $\SC^{P,Q}_{\eps}$ is as defined in Theorem~\ref{thm:private-hypothesis-testing}.
\end{lemma}
\begin{proof}
We describe the reduction for the offline version of the problem, the reduction in the online setting is identical.
The reduction is easy to describe.
We partition the sample indices into intervals of length $\SC^{P,Q}_\eps$.
More precisely, let $Y_j = \{x_{(i-1)\SC^{P,Q}_\eps + 1}, \dots, x_{i\SC^{P,Q}_\eps} \}$, for $j = 1$ to $\lfloor n/\SC^{P,Q}_\eps \rfloor$, and disregard the remaining ``tail'' of $x_i$'s. 
We run the algorithm of Theorem~\ref{thm:private-hypothesis-testing} on each $Y_j$, and produce a bit $z_j=1$ if the algorithm outputs that the distribution is $P$, and a $z_j = -1$ otherwise.

We show that this forms a valid instance of the change-point detection problem in the lemma statement.
Let $k^\ast$ be the change-point in the original problem, and suppose it belongs to $Y_{\tilde k^\ast}$.
For every $j < \tilde k^\ast$, all the samples are from $P$, and by Theorem~\ref{thm:private-hypothesis-testing}, each $z_j$ will independently be $1$ with probability $\tau_0 \geq 2/3$.
Similarly for every $j > \tilde k^\ast$, all the samples are from $Q$, and each $z_j$ will independently be $-1$ with probability at least $1 - \tau_1 \geq 2/3$.

Finally, we show that the existence of an $(\alpha,\beta)$-accurate algorithm that solves this problem also solves the original problem.
Suppose that the output of the algorithm on the restricted change-point detection problem is $j$.
To map this to an answer to the original problem, we let $\hat k = (j-1)\SC^{P,Q}_\eps$.

First, note that $\hat k$ will be $\ve$-differentially private.
We claim that the sequence of $z_j$'s is $\ve$-differentially private. 
This is because the algorithm of Theorem~\ref{thm:private-hypothesis-testing} is $\ve$-differentially private, we apply the algorithm independently to each component of the partition, and each data point can only affect one component (since they are disjoint).
Privacy of $\hat k$ follows since privacy is closed under post-processing.

Finally, we show the accuracy guarantee.
In the restricted change-point detection problem, with probability at least $1-\beta$, the output $j$ will be such that $|j - \tilde k^\ast| \leq \alpha$.
In the original problem's domain, this corresponds to a $\hat k$ such that $|\hat k - k^\ast| \leq (\alpha+1)\SC^{P,Q}_\eps$, as desired.
\end{proof}

\subsection{Solving Bernoulli Change-Point Detection}
\label{sec:bernoulli-cpd}
In this section, we show that there is a $(\Theta(\log(1/\beta)) + 1, \beta)$-accurate algorithm for the restricted change-point detection problem.
Combined with Lemma~\ref{lem:reduction}, this implies Theorem~\ref{thm:cpd}.
\begin{lemma}\label{lem:bernoulli-cpd}
There exists an efficient $(O(\log(1/\beta), \beta)$-accurate algorithm for the offline restricted change-point detection problem (as defined in Lemma~\ref{lem:reduction}).

Similarly, there exists an efficient $(O(\log(1/\beta), \beta)$-accurate algorithm for the online restricted change-point detection problem.
This algorithm requires as input a value $n$ such that $n = \Omega\left(\log\left(\frac{k^\ast}{n \beta}\right)\right)$.
If the algorithm is accurate, it will observe at most $k^\ast + 2n$ data points, and with high probability observe $k^\ast + O(n \log n)$ data points.
\end{lemma}
\begin{proof}
We start by describing the algorithm for the offline version of the problem.
We then discuss how to reduce from the online problem to the offline problem. 

\paragraph{Offline Change-Point Detection.}
We define the function 
\[\ell(t) = \sum_{j=t}^n z_j. \]
The algorithm's output will be $\hat k = \arg\min_{1 \leq t \leq n} \ell(t)$.

Let $k^\ast$ be the true change-point index.
To prove correctness of this algorithm, we show that $\ell(k^\ast+1) - \ell(t) < 0$ for all $t \geq k^\ast + 1 + \Theta(\log(1/\beta))$, and that $\ell(k^\ast-1) - \ell(t) < 0$ for all $t \leq k^\ast - 1 - \Theta(\log(1/\beta))$. 
Together, these will show that $\arg\min_{1 \leq t \leq n} \ell(t) \in [k^\ast - 1 - \Theta(\log(1/\beta)), k^\ast + 1 + \Theta(\log (1/\beta))]$, proving the result.
For the remainder of this proof, we focus on the former case, the latter will follow symmetrically.
Specifically, we will show that $\ell(k^\ast+1) - \ell(t) < 0$ for all $t \geq k^\ast + 1 + c\log(1/\beta)$, where $c > 0$ is some large absolute constant.

Observe that for $t \geq k^\ast+1$,
\[
\ell(k^\ast+1) - \ell(t) = \sum_{j=k^\ast+1}^{t-1} z_j
\]
forms a biased random walk which takes a $+1$-step with probability $\tau_1 \leq 1/3$ and a $-1$-step with probability $1 - \tau_1 \geq 2/3$.
Define $M_i = \ell(k^\ast+1) - \ell(k^\ast+1 + i) + i(1 - 2 \tau_1)$ for $i = 0$ to $n - k^\ast - 1$, and note that this forms a martingale sequence.
We will use Theorem 4 of~\cite{Balsubramani15}, which provides a finite-time law of the iterated logarithm result.
Specialized to our setting, we obtain the following maximal inequality, bounding how far this martingale deviates away from $0$.
\begin{theorem}[Follows from Theorem 4 of~\cite{Balsubramani15}]
Let $c > 0$ be some absolute constant. 
With probability at least $1 - \beta$, for all $i \geq c \log(1/\beta)$ simultaneously,
\[|M_i| \leq O\left(\sqrt{i\left(\log\log(i) + \log(1/\beta)\right)}\right). \]
\end{theorem}

This implies that, with probability at least $1 - \beta$, we have that for all $i \geq c \log(1/\beta)$ 
\[ \ell(k^\ast+1) - \ell(k^\ast + 1 + i) = M_i - i(1 - 2\tau_1) \leq O\left(\sqrt{i \left(\log \log (i) + \log(1/\beta)\right)}\right) - \frac{i}{3}.\]
Note that the right-hand side is non-increasing in $i$, so it is maximized at $i = c \log(1/\beta)$, and thus
\[ \ell(k^\ast+1) - \ell(k^\ast + 1 + i) \leq  O\left(\sqrt{\log(1/\beta) \left(\log \log \log(1/\beta) + \log(1/\beta)\right)}\right) - \frac{c \log(1/\beta)}{3} < 0,\]
where the last inequality follows for a sufficiently large choice of $c$.

\paragraph{Online Change-Point Detection.}
The algorithm will be as follows.
Partition the stream into consecutive intervals of length $n$, which we will draw in batches.
If an interval has more $-1$'s than $+1$'s, then call the offline change-point detection algorithm on the final $2n$ data points with failure probability parameter set to $\beta/4$, and output whatever it says.

Let $k^\ast$ be the true change-point index.
First, we show that with probability $\geq 1 - \beta/4$, the algorithm will not see more $-1$'s than $+1$'s in any interval before the one containing $k^\ast$.
The number of $+1$'s in this interval will be distributed as $\operatorname{Binomial}(n,\tau_0)$ for $\tau_0 > 2/3$.
By a Chernoff bound, the probability that we have $> n/2$ $-1$'s is at most $\exp(-\Theta(n))$.
Taking a union bound over all $O\left(\frac{k^\ast}{n}\right)$ intervals before the change point gives a failure probability of $\frac{k^\ast}{n}\exp(-\Theta(n)) \leq \beta/4$, where the last inequality follows by our condition on $n$.

Next, note that the interval following the one containing $k^\ast$ will have a number of $+1$'s which is distributed as $\operatorname{Binomial}(n,\tau_1)$ for $\tau_1 < 1/3$.
By a similar Chernoff bound as before, the probability that we have $> n/2$ $+1$'s is at most $\exp(-\Theta(n)) \ll \beta/4$.
Therefore, with probability $1 - \beta/2$, the algorithm will call the offline change-point detection algorithm on an interval containing the true change point $k^\ast$.

We conclude by the correctness guarantees of the offline change-point detection algorithm.
Note that we chose the failure probablity parameter to be $\beta/4$, as the offline algorithm may either be called at the interval containing $k^\ast$, or the following one, and we take a union bound over both of them.
\end{proof}

\subsection{Private Goodness-of-Fit Change-Point Detection}
\label{sec:identity-cpd}
Our reduction as given above is rather general: and it can apply to more general change-point detection settings than those described above.
For instance, the above discussion assumes we know both the initial and final distributions $P$ and $Q$.
Instead, one could imagine a setting where one knows the initial distribution $P$ but not the final distribution $Q$, which we term goodness-of-fit change-point detection.

\begin{definition}
In the \emph{offline $\gamma$-goodness-of-fit change-point detection problem}, we are given a distribution $P$ over domain $\cX$ and a data set $X = \{x_1, \dots, x_n\}$.
We are guaranteed that there exists $k^\ast \in [n]$ such that $x_1, \dots, x_{k^\ast-1} \sim P$, and $x_{k^\ast}, \dots, x_n \sim Q$, for some fixed (but unknown) distribution $Q$ over domain $\cX$, such that $\tvd(P,Q) \geq \gamma$.
The goal is to output $\hat k$ such that $|\hat k - k^\ast|$ is small.
\end{definition}
We note that analogous definitions and results hold for the online version of this problem, as in the previous sections.

We omit the full details of the proof, but it proceeds by a very similar argument to that in Sections~\ref{sec:reduction-cpd} and~\ref{sec:bernoulli-cpd}.
In particular, it is possible to prove an analogue of Lemma~\ref{lem:reduction}, at which point we can apply Lemma~\ref{lem:bernoulli-cpd}.
The only difference is that we need an algorithm for private goodness-of-fit testing, rather than Theorem~\ref{thm:private-hypothesis-testing} for hypothesis testing.
We use the following result from~\cite{AcharyaSZ18a}.
\begin{theorem}[Theorem 13 of~\cite{AcharyaSZ18a}]
Let $P$ be a known distribution over $\cX$, and let $\mathcal{Q}$ be the set of all distributions $Q$ over $\cX$ such that $\tvd(P,Q) \geq \gamma$.
Given $n$ samples from an unknown distribution which is either $P$, or some $Q \in \mathcal{Q}$, there exists an efficient $\ve$-differentially private algorithm which distinguishes between the two cases with probability $\geq 2/3$ when
$$n = \Theta\left(\frac{|\cX|^{1/2}}{\gamma^2} + \frac{|\cX|^{1/2}}{\gamma \ve^{1/2}} + \frac{|\cX|^{1/3}}{\gamma^{4/3}\ve^{2/3}} + \frac{1}{\gamma \ve}\right).$$
\end{theorem}

With this in hand, we have the following result for goodness-of-fit changepoint detection.
\begin{theorem}
There exists an efficient $\ve$-differentially private and $(\alpha, \beta)$-accurate algorithm for offline $\gamma$-goodness-of-fit change-point detection with 
$$\alpha = \Theta\left(\left(\frac{|\cX|^{1/2}}{\gamma^2} + \frac{|\cX|^{1/2}}{\gamma \ve^{1/2}} + \frac{|\cX|^{1/3}}{\gamma^{4/3}\ve^{2/3}} + \frac{1}{\gamma \ve}\right) \cdot \log(1/\beta)\right).$$
\end{theorem}

\section*{Acknowledgments}
We are grateful to Salil Vadhan for valuable discussions in the early stages of this work.

{\small
\bibliographystyle{alpha}
\bibliography{biblio}
}

\appendix
\section{Glossary of Tests}
\label{sec:glossary}
\noindent In this section, we list all the tests mentioned in this paper.
As some mnemonics, the prefix ``s'' indicates that a test is ``soft,'' meaning that the test's output is distributed as Bernoulli random variable with parameter proportional to some test statistic.
The prefix ``n'' means that a statistic is ``noisy,'' which we enforce by adding Laplace noise. 
The prefix ``c'' means that a statistic is ``clamped'': to limit the sensitivity of the statistic, we clamp the value of each summand to a fixed range, so that terms can not be unboundedly large.
\smallskip

\noindent The $\LLR$ is the log-likelihood ratio statistic:
\begin{equation*}
\LLR(x_1,\dots,x_n) = \sum_{i=1}^{n} \log \frac{P(x_i)}{Q(x_i)}.
\end{equation*}

\noindent The $\sLLR$ is the soft log-likelihood ratio test:
\[
\sLLR(x_1, \cdots, x_n) =  
\begin{cases}
P & \text{ with probability } g(x)\\
Q & \text{ with probability } 1-g(x)
\end{cases}
\]
where 
\[
    g(x)=\frac{\exp(\frac{1}{2}\log\frac{P(x)}{Q(x)})}{1+\exp(\frac{1}{2}\log\frac{P(x)}{Q(x)})}\in[0,1].
\]

\noindent The $\CLLR$ is the clamped log-likelihood ratio statistic:
\begin{equation*}
\CLLR_{a,b}(x_1, \cdots, x_n) = \sum_i \left[ \log \frac{P(x_i)}{Q(x_i)} \right]_{a}^{b}.
\end{equation*}

\noindent The $\sCLLR$ is the soft clamped log-likelihood ratio test:
\begin{equation*}
\sCLLR_{a,b}(x_1, \cdots, x_n)= \begin{cases} 
P & \text{with probability } \propto \exp(\frac12 \CLLR_{a,b}(x))\\ 
Q & \text{with probability }  \propto 1
\end{cases}
\end{equation*}

\noindent The $\NCLLR$ is the noisy log-likelihood ratio test, which is also $\eps$-differentially private:
\begin{equation*}
  \NCLLR_{-\eps',\eps}(x_1, \cdots, x_n) = \sum_i \left[ \log \frac{P(x_i)}{Q(x_i)} \right]_{-\eps'}^{\eps} + \mathrm{Lap}(2).
\end{equation*}

\section{Proof of existence of $\eps'$}
\label{sec:eps-proof}

\begin{lemma}
Let $P$ and $Q$ be two arbitrary distributions and  let $\tau$ be as defined in \eqref{def:tau}. Assume without loss of generality that $\tau=\int_{\cX}\max\{P(x)-e^{\eps}Q(x),0\}\dx$. Then there exists $\eps'\in[0,\eps]$ such that $\tau=\int_{\cX}\max\{Q(x)-e^{\eps'}P(x),0\}\dx$.
\end{lemma}

\begin{proof}
First, we claim that the function $\phi\colon[0,\infty)\to \R$ defined by $\phi(y)=\int_{\cX}\max\{Q(x)-e^{y}P(x),0\}\dx$ is continuous. Indeed, define, for $y\geq 0$, the function $\phi_y(x) = \max\{Q(x)-e^{y}P(x),0\}$ (on $\cX$). For any fixed $y_0\geq 0$, we have (i)~pointwise convergence of $\phi_y$ to $\phi_{y_0}$; (ii)~(Lebesgue) integrability of $\phi_y$ on $\cX$; and (iii)~for all $y\geq 0$ and $x\in\cX$, $|\phi_y(x)| \leq Q(x)$ with $\int_{\cX} Q(x)\dx < \infty$. By the dominated convergence theorem, we get that $\lim_{y\to y_0}\phi(y) = \phi(y_0)$.

Now, by assumption, $\phi(\eps)=\int_{\cX}\max\{Q(x)-e^{\eps}P(x),0\}\dx\le\tau$ and $\phi(0)=\int_{\cX}\max\{Q(x)-e^{0}P(x),0\}\dx=\tvd(P,Q)\ge\tau$. Thus, by the intermediate value theorem, there exists $\eps'\in[0,\eps]$ such that $\phi(\eps') = \tau$.
\end{proof}

\section{Proof of Theorem~\ref{thm:private-hypothesis-testing} when $P$ and $Q$ have disjoint support.}\label{app:disjoint:supports}

Suppose that $P$ and $Q$ have disjoint support so $\tilde{P}=\tilde{Q}=0$. Then Theorem \ref{thm:private-hypothesis-testing} reduces to:

\begin{theorem}
  If $P$ and $Q$ have disjoint support then $\sCLLR_{-\eps',\eps}$ and $\NCLLR_{-\eps',\eps}$ are asymptotically optimal among all $\eps$-DP tests and have sample complexity $SC^{P,Q}(\sCLLR_{-\eps',\eps})= SC^{P,Q}(\NCLLR_{-\eps',\eps}) = \Theta\left(\frac{1}{\eps}\right)$.
\end{theorem}
\begin{proof}
First consider $\CLLR_{-\eps',\eps}$, and note that the disjoint supports imply that $\tau = 1$ and $\eps' = \eps$.
Under the distribution $P$, $\CLLR_{-\eps, \eps}$ will (deterministically) be $n\eps$. Indeed, the test $\NCLLR_{-\eps,\eps}$ will be distributed as $n\eps + \Lap(2)$, and for this to be greater than $0$ (and thus correct) with probability $\geq 2/3$, it is necessary and sufficient that $n = \Omega(1/\eps)$.
Similarly, $\sCLLR_{-\eps, \eps}$ is correct with probability $\frac{\exp\left(\frac12 n\eps\right)}{1 + \exp\left(\frac12 n\eps\right)} \geq 2/3$ as long as $n = \Omega(1/\eps)$.
A symmetric argument shows correctness under the distribution $Q$.

For the lower bound, consider the coupling of $P^n$ and $Q^n$ that takes a sample $X\sim P^n$ and simply resamples it from $Q^n$. Since $P^n$ and $Q^n$ have disjoint support, $d_{\eps}(X,Y)=\min\{\eps n, 1\}$. Thus, by Lemma \ref{lem:dp->wasserstein}, we have that $n\ge \Omega(1/\eps)$ is a necessary condition.
\end{proof}
\section{Proof of Lemma~\ref{lem:vargap}}\label{vargapproof}

\begin{proof}
  \newcommand{\threshold}{t}
  Suppose, without loss of generality, that $\mathbb{E}_{X \sim
    P^{n'}}[S(X)]>\mathbb{E}_{X \sim Q^{n'}}[S(X)]$, let $n' =
  12c^2 n$ and let $\threshold$ be equal to $\frac{1}{2}\left(\mathbb{E}_{X \sim P^{n'}}[S(X)] + \mathbb{E}_{X \sim Q^{n'}}[S(X)]\right)$. Note also that since $S(X)=\sum_{i=1}^n f(x_i)$, we have that $|\mathbb{E}_{X \sim P^{n'}}[S(X)]-\mathbb{E}_{X \sim Q^{n'}}[S(X)]| = 12c^2|\mathbb{E}_{X \sim P^{n}}[S(X)]-\mathbb{E}_{X \sim Q^{n}}[S(X)]|$ and $\var{X \sim P^{n'}}{S(X)} = 12c^2\var{X \sim P^n}{S(X)}$. 
 Then 
\begin{align*}
  \pr{X \sim P^{n'}}{S(X)<\threshold}&\le \pr{X \sim P^{n'}}{|S(X)-\mathbb{E}_{X \sim P^{n'}}[S(X)]|>\frac{|\mathbb{E}_{X \sim P^{n'}}[S(X)]-\mathbb{E}_{X \sim Q^{n'}}[S(X)]|}{2}}\\\
&\le 4\frac{\var{X \sim P^{n'}}{S(X)}}{|\mathbb{E}_{X \sim P^{n'}}[S(X)]-\mathbb{E}_{X \sim Q^{n'}}[S(X)]|^2}\\
&\le 4\frac{12c^2c^2}{144c^4} = \frac{1}{3}
\end{align*}
The second inequality uses Chebyshev's inequality, and the third uses the assumptions of the lemma statement.
The proof that $\pr{Q^{n'}}{S(X)>\threshold}\le \frac{1}{3}$
proceeds similarly.


\end{proof}

\section{The Advantage of $\sCLLR$}
\label{sec:adv-scllr}
In this section, we show that the Soft Clamped Log-Likelihood Ratio test ($\sCLLR$) achieves advantage related to the quantities $\tau$ and transformed distributions $P',Q'$ introduced earlier.  Recall the definition
\begin{equation*}
\sCLLR_{-\eps',\eps}(x)= \begin{cases} 
P & \text{with probability } \propto \exp(\frac12 \CLLR_{-\eps',\eps}(x))\\ 
Q & \text{with probability }  \propto 1
\end{cases}
\end{equation*}
Also recall the definitions of $\tilde{P},\tilde{Q},P',Q',\eps'$ from Section~\ref{sec:results}.  We now have the following lemma.
\begin{lem}
For any $P$ and $Q$, 
$
\advg^{P,Q}_{1}(\sCLLR_{-\eps',\eps}) = \Theta\left(\eps\tau+H^2(P',Q')\right)
$
\end{lem}
To ease the notation, for the remainder of this section we will simply write $\sCLLR = \sCLLR_{-\eps',\eps}$ and $\advg(\sCLLR) = \advg^{P,Q}_{1}(\sCLLR)$.
\begin{proof}
Firstly, by definition we have 
\[
\advg(\sCLLR)=\frac{1}{2}\int_{}(P(x)-Q(x))\frac{e^{\CLLR(x)/2}-1}{e^{\CLLR(x)/2}+1}dx,
\] 
where $\CLLR(x) = \CLLR_{-\eps',\eps} =\left[\log\frac{P(x)}{Q(x)}\right]_{-\eps'}^{\eps}=\log\frac{\tilde{P}(x)}{\tilde{Q}(x)}.$  The latter equality holds by construction of $\tilde{P},\tilde{Q}$. Now we can break up the integrand $(P(x)-Q(x))\frac{e^{\CLLR(x)/2}-1}{e^{\CLLR(x)/2}+1}$ as
\begin{align*}
&\underbrace{(P(x)-\tilde{P}(x))\frac{e^{\CLLR(x)/2}-1}{e^{\CLLR(x)/2}+1}}_{\text{term 1}} + \underbrace{(\tilde{Q}(x)-Q(x))\frac{e^{\CLLR(x)/2}-1}{e^{\CLLR(x)/2}+1}}_{\text{term 2}} + \underbrace{(\tilde{P}(x)-\tilde{Q}(x))\frac{e^{\CLLR(x)/2}-1}{e^{\CLLR(x)/2}+1}}_{\text{term 3}} \,. 
\end{align*}
Now, $P$ and $\tilde{P}$ only differ on the set 
\[
\Set{S} = \setOf{x}{P(x)-e^{\eps}Q(x)>0}
\]
so the term 1 is only non-zero if $x\in \Set{S}$, where it takes the value
\[
(P(x)-\tilde{P}(x))\frac{e^{\eps/2}-1}{e^{\eps/2}+1}.
\]
Similarly, $Q$ and $\tilde{Q}$ only differ on the set 
\[
\Set{T} = \setOf{x}{Q(x)-e^{\eps'}P(x)>0}
\]
so term 2 is only non-zero if $x\in \Set{T}$, where it takes the value 
\[
(\tilde{Q}(x)-{Q}(x))\frac{e^{-\eps'/2}-1}{e^{-\eps'/2}+1}=(Q(x)-\tilde{Q}(x))\frac{e^{\eps'/2}-1}{e^{\eps'/2}+1}.
\]
Overall, we get the following equality: 
 \begin{align*}
 (P(x)-Q(x)) \frac{e^{\CLLR(x)/2}-1}{e^{\CLLR(x)/2}+1} 
 &= (P(x)-\tilde{P}(x)) \frac{e^{\eps/2}-1}{e^{\eps/2}+1} + (Q(x)-\tilde{Q}(x)) \frac{e^{\eps'/2}-1}{e^{\eps'/2}+1} + (\tilde{P}(x)-\tilde{Q}(x)) \frac{e^{\CLLR(x)/2}-1}{e^{\CLLR(x)/2}+1}\,.
\end{align*}
Integrating over $x$, we get that $2 \cdot \advg(\sCLLR)$ is equal to
 \begin{align}
\frac{e^{\eps/2}-1}{e^{\eps/2}+1}  \underbrace{\int (P(x)-\tilde{P}(x))dx}_\tau + \frac{e^{\eps'/2}-1}{e^{\eps'/2}+1} \underbrace{\int (Q(x)-\tilde{Q}(x)) dx}_\tau  + \int (\tilde{P}(x) - \tilde{Q}(x)) \frac{e^{\log\frac{\tilde{P}(x)}{\tilde{Q}(x)}/2}-1}{e^{\log\frac{\tilde{P}(x)}{\tilde{Q}(x)}/2}+1} dx
     \label{eq:swLLR-advantage}
 \end{align}
Now, since $P',Q'$ are simply renormalizations of $\tilde{P}$, $\tilde{Q}$ by the same factor $1-\tau$,
\begin{align*}
 &\frac{e^{\frac{1}{2}\log\frac{\tilde{P}(x)}{\tilde{Q}(x)}}-1}{e^{\frac{1}{2}\log\frac{\tilde{P}(x)}{\tilde{Q}(x)}}+1} 
 ={} \frac{e^{\frac{1}{2}\log\frac{P'(x)}{Q'(x)}}-1}{e^{\frac{1}{2}\log\frac{P'(x)}{Q'(x)}}+1}  
 ={} \frac{\sqrt{\frac{P'(x)}{Q'(x)}} -1}{\sqrt{\frac{P'(x)}{Q'(x)}}+1 } 
 ={} \frac{\sqrt{P'(x)}- \sqrt{Q'(x)}}{\sqrt{P'(x)} + \sqrt{Q'(x)}}  
 ={} \frac{\left(\sqrt{P'(x)}- \sqrt{Q'(x)}\right)^2}{P'(x)-Q'(x)}.
 \end{align*}
 We can therefore simplify the expression in \eqref{eq:swLLR-advantage} to obtain: 
\[
2\cdot \advg(\sCLLR) =  \left(\frac{e^{\eps/2}-1}{e^{\eps/2}+1} + \frac{e^{\eps'/2}-1}{e^{\eps'/2}+1}\right)\cdot \tau + H^2(P',Q')=\Theta(\eps\tau)+H^2(P',Q'),
\]
 as desired.
\end{proof}

\section{Proof of Lemma \ref{lem:wasserstein->dp}}
\label{sec:wasserstein->dp}

The proof of Lemma~\ref{lem:wasserstein->dp} relies on the following dual formulation of Wasserstein distance:
\begin{lem} \label{lem:wassersteindual}
For any metric $d \from \cX \times \cX \to \R_{\geq 0}$, $W_{d}(P,Q) \geq \alpha$ if and only if there exists a function $f\from \cX \to \R_{\geq 0}$ such that
\begin{enumerate}
\item $f$ is $1$-Lipschitz with respect to $d$, and
\item $\ex{}{f(P)} - \ex{}{f(Q)} \geq \alpha$, and
\item $\max_{X \in \cX} f(X) \leq 2 \cdot \max_{X,Y \in \cX} d(X,Y)$.
\end{enumerate}
\end{lem}

\noindent Using Lemma~\ref{lem:wassersteindual}, there exists a $1$-Lipschitz function $f \from \cX \to [0,2]$ 
such that $\ex{}{f(P)} - \ex{}{f(Q)} \geq W_{d_{\eps}}(P,Q)$.  If we define the algorithm
\begin{equation*}
M(X) = 
\begin{cases}
P & \textrm{with probability $\frac{1+ f(X)}{4}$}\\
Q & \textrm{with probability $\frac{3-f(X)}{4}$}
\end{cases}
\end{equation*}
then
\begin{enumerate}
\item $M$ satisfies $\eps$-DP, and
\item $\left|\pr{X \sim P}{M(X)=P}-\pr{X \sim Q}{M(X)=P}\right| = \frac{1}{4} W_{d_{\eps}}(P,Q)$.
\end{enumerate}
The second claim is immediate from the definition of $M$.  To see that $M$ is $\eps$-DP, note that any function that is $1$-Lipschitz with respect to $d_\eps$ is $\eps$-Lipschitz with respect to 
$d_H$. For any pair of neighbours $X$ and $Y$, we have
\begin{equation*}
\frac{\pr{M}{M(X) = P}}{\pr{M}{M(Y)=P}} \leq \frac{\frac{1}{4} + \frac{1}{4}f(X)}{\frac{1}{4} + \frac{1}{4}f(Y)} \leq \frac{\frac{1}{4} + \frac{\eps}{4}}{\frac{1}{4}} = 1+\eps \leq e^{\eps}
\end{equation*}
and the same holds for $\pr{M}{M(X) = Q}$.


\end{document}